\documentclass[11pt]{article}
\usepackage{graphicx,amsmath,amsfonts,amscd,amssymb,bm,cite,epsfig,epsf,url,color,mathbbol
}
\usepackage{fullpage} \usepackage[small,bf]{caption}
\newtheorem{theorem}{Theorem}[section]
\newtheorem{lemma}[theorem]{Lemma}

\newtheorem{definition}[theorem]{Definition}

\newcommand{\onetwonorm}[1]{||#1||_{1,2}}
\newcommand{\opnorm}[1]{\|#1\|}
\newcommand{\fronorm}[1]{\|#1\|_{F}}

\newcommand{\twonorm}[1]{\|#1\|_{\ell_2}}

\newcommand{\nucnorm}[1]{||#1||_*}
\newcommand{\cA}{\mathcal{A}}

\newcommand{\R}{\mathbb{R}}

\newcommand{\<}{\langle}
\renewcommand{\>}{\rangle}

\newcommand{\goto}{\rightarrow}

\renewcommand{\P}{\operatorname{\mathbb{P}}}
\newcommand{\E}{\operatorname{\mathbb{E}}}

\newcommand{\rank}{\operatorname{rank}}

\newcommand{\trace}{\operatorname{trace}}

\numberwithin{equation}{section}

\def \endprf{\hfill {\vrule height6pt width6pt depth0pt}\medskip}
\newenvironment{proof}{\noindent {\bf Proof} }{\endprf\par}

\title{Tight oracle bounds for low-rank matrix recovery\\
  from a minimal number of random measurements}
\author{Emmanuel J. Cand\`es$^{1,2}$ and Yaniv Plan$^{3}$\thanks{Corresponding author: Yaniv Plan. Email: yanivplan@gmail.com}\\
  \vspace{-.1cm}\\
  $^1$Department of Statistics, Stanford University, Stanford, CA 94305\\
  \vspace{-.3cm}\\
  $^2$Department of Mathematics, Stanford University, Stanford, CA 94305\\
  \vspace{-.3cm}\\
  $^3$Applied and Computational Mathematics, Caltech, Pasadena, CA
  91125}

\date{December 2009}

\begin{document}

\maketitle

\vspace{-0.3in}

\begin{abstract}
  This paper presents several novel theoretical results regarding the
  recovery of a low-rank matrix from just a few measurements
  consisting of linear combinations of the matrix entries.  We show
  that properly constrained nuclear-norm minimization stably recovers
  a low-rank matrix from a constant number of noisy measurements per
  degree of freedom; this seems to be the first result of this nature.
  Further, the recovery error from noisy data is within a constant of
  three targets: 1) the minimax risk, 2) an `oracle' error that would be available if the column space of the matrix were known, and 3) a more adaptive `oracle' error which would be available with the knowledge of the column space corresponding to the part of the matrix that stands above the noise. Lastly, the error bounds regarding
  low-rank matrices are extended to provide an error bound when the
  matrix has full rank with decaying singular values.  The analysis in
  this paper is based on the restricted isometry property (RIP)
  introduced in \cite{CT05} for vectors, and in \cite{Recht07} for
  matrices.
\end{abstract}


\maketitle

\section{Introduction}
Low-rank matrix recovery is a burgeoning topic drawing the attention
of many researchers in the closely related field of sparse
approximation and compressive sensing.  To draw an analogy, in the
sparse approximation setup, the signal $y$ is modeled as a sparse
linear combination of elements from a dictionary $D$ so that
\[
y = Dx, 
\] 
where $x$ is a sparse coefficient vector.  The goal is to recover $x$.
In the matrix recovery problem, the signal to be recovered is a
low-rank matrix $M \in \R^{n_1 \times n_2}$, about which we have
information supplied by means of a linear operator $\cA:\R^{n_1 \times
  n_2} \rightarrow \R^m$ (typically, $m$ is far less than $n_1 n_2$), 
\[
y = \cA(M), \quad y \in \R^m. 
\]
In both cases, signal recovery appears to be an ill-posed problem
because there are many more unknowns than equations.  However, as has
been shown extensively in the sparse-approximation literature, the
assumption that the object of interest is sparse makes this problem
meaningful even when the linear system of equations is apparently
underdetermined.  Further, when the measurements are corrupted by
noise, we now know that by taking into account the parsimony of the
model, one can insure that the recovery error is within a log factor
of the error one would achieve by regressing $y$ onto the
low-dimensional subspace spanned by those columns with $x_i \neq 0$;
the squared error is {\em adaptive}, and proportional to the true
dimension of the signal \cite{DS, Bickel07}.

In this paper, we derive similar results for matrix recovery. In
contrast to results available in the literature on compressive sensing
or sparse regression, we show that the error bound is within a
constant factor (rather than a log factor) of an idealized `oracle'
error bound achieved by projecting the data onto a smaller subspace
given by the `oracle' (and also within a constant of the minimax error
bound).  This error bound also applies to full-rank matrices (which
are well-approximated by low-rank matrices), and there appears to be
no analogue of this in the compressive sensing world.

Another contribution of this paper is to lower the number of
measurements to stably recover a matrix of rank $r$ by convex
programming. It is not hard to see that we need at least $m \ge
(n_1+n_2-r)r$ measurements to recover matrices of rank $r$, by any
method whatsoever. To be sure, if $m < (n_1+n_2-r)r$, we will always
have two distinct matrices $M$ and $M'$ or rank at most $r$ with the
property $\cA(M) = \cA(M')$ no matter what $\cA$ is.  To see this, fix
two matrices $U \in \R^{n_1 \times r}$, $V \in \R^{n_2 \times r}$ with
orthonormal columns, and consider the linear space of matrices of the
form
\[
T = \{UX^* - YV^* : X \in \R^{n_2 \times r}, Y \in \R^{n_1 \times r}\}.
\]
The dimension of $T$ is $r(n_1+n_2-r)$. Thus, if $m < (n_1+n_2-r)r$,
there exists $M = UX^* - YV^* \neq 0$ in $T$ such that $\cA(M) =
0$. 
This proves the claim since $\cA(UX^*) = \cA(YV^*)$ for two distinct
matrices of rank at most $r$.
Now a novel result of this paper is that, even without knowing that
$M \in T$, one can stably recover $M$ from a constant times $(n_1+
n_2)r$ measurements via nuclear-norm minimization.  Once again, in
contrast to similar results in compressive sensing, the number of
measurements required is within a constant of the theoretical lower
limit -- there is no extra log factor.

\subsection{A few applications}
Following a series of advances in the theory of low-rank matrix
recovery from undersampled linear measurements \cite{Recht07, CR08,
  CT09, CP09, MontanariNoiseless, MontanariNoisy, Lee2009, Ma09,
  Jain09}, a number of new applications have sprung up to join ranks
with the already established ones. A quick survey shows that low-rank
modeling is getting very popular in science and engineering, and we
present a few eclectic examples to illustrate this point.
\begin{itemize}
\item \textbf{Quantum state tomography \cite{Gross09}.} In quantum
  state tomography, a mixed quantum state is represented as a square
  positive semidefinite matrix, $M$ (with trace 1).  If $M$ is actually a pure
  state, then it has rank 1, and more generally, if it is
  approximately pure then it will be well approximated by a low-rank
  matrix \cite{Gross09}.
\item \textbf{Face recognition \cite{Ma09, Basri03}.} Here the
  sequence of signals $\{y_i\}$ are images of the same face under
  varying illumination.  In theory and under idealized circumstances
  (the images are assumed to be convex, Lambertian objects), these
  faces all reside near the same nine-dimensional linear subspace
  \cite{Basri03}.  In practice, face-recognition techniques based on
  the assumption that these images reside in a low-dimensional
  subspace are highly successful \cite{Basri03, Ma09}.
\item \textbf{Distance measurements.}  Let $x_i \in \R^d$ be a
  sequence of vectors representing several positions in $d$
  dimensional space, and let $M$ be the matrix of squared distances
  between these vectors, $M_{i,j} = \twonorm{x_i - x_j}^2$.  Then $M$
  has rank bounded by $d+2$.  To see this, let $X = [x_1, x_2, \ldots,
  x_n]$ be a concatenation of the positions vectors.  Then, letting
  $\{e_i\}$ be the standard basis vectors,
\[
M_{i,j} = (e_i - e_j)^* X^* X (e_i - e_j) = -2e_i^* X^* X e_j + e_i^*
\mathbb{1} q^* e_j + e_i^* q \mathbb{1}^* e_j, 
\] 
where $\mathbb{1}$ is a vector containing all ones, and $q$ is a
vector with $q_i = \<x_i, x_i\>$.  Thus $M = X^*X + \mathbb{1}q^* +
q\mathbb{1}^*$.  The first matrix has rank bounded by $d$ and the
second two have rank bounded by 1.  In fact, one can project out
$\mathbb{1}q^* + q\mathbb{1}^*$ in order to reduce the rank to $d$,
which in usual applications will be 2 for positions constrained to lie
in the plane, and 3 for positions constrained to lie somewhere in
space.
\end{itemize}

Quantum state tomography lends itself perfectly to the compressive
sensing framework. On an abstract level, one sees measurements
consisting of linear combinations of the unknown quantum state $M$ --
inner products with certain observables which can be chosen with some
flexibility by the physicist -- and the goal is to recover a good
approximation of $M$.  The size of $M$ grows exponentially with the
number of particles in the system, so one would like to use the
structure of $M$ to reduce the number of measurements required, thus
necessitating compressive sensing (see \cite{Gross09} for a more in
depth discussion and a specific analysis of this problem). 
\footnote{
An interesting point about quantum state tomography is that if one enforces the constraints $\trace(M) = 1$ and $M\succeq 0$ then this ensures that $\nucnorm{M} = 1$, and the scientist is left with a feasibility problem.  In \cite{Gross09} the authors suggest to solve this feasibility problem by removing a constraint and then performing nuclear-norm minimization and they show that under certain conditions this is sufficient for exact recovery (and thus of course the solution obeys the unenforced constraint).} 
Another
more established example is sensor localization, in which one sees a
subset of the entries of a distance matrix because the sensors have
low power and can only sense reliably its distance to nearby sensors.
The goal is to fill in the missing entries (matrix completion).  In
some applications of the face recognition example, one would see the
entire set of faces (the sampling operator is the identity), and the
low-rank structure can be used to remove sparse errors, but otherwise
arbitrarily gross, from the data as described in \cite{Ma09} (we
include this example to illustrate the different uses of the low-rank
matrix model, but also note that it is quite different than the
problem addressed in our paper).

\subsection{Prior literature}
There has recently been an explosion of literature regarding low-rank
matrix recovery, with special attention given to the matrix completion
subproblem (as made famous by the million dollar Netflix Prize).
Several different algorithms have been proposed, with many drawing
their roots from standard compressive sensing techniques
\cite{Recht07, CR08, CP09, MontanariNoisy, Ma08, SVT, Dai2009,
  Lee2009, Jain09}.  For example, nuclear-norm minimization is highly
analogous to $\ell_1$ minimization (as a convex relaxation to an
intractable problem), and the algorithms analyzed in this paper are
analogous to the Dantzig Selector and the LASSO.

The theory regarding the power of nuclear-norm minimization in
recovering low-rank matrices from undersampled measurements began
with a paper by Recht et al. \cite{Recht07}, which sought to bridge
compressive-sensing with low-rank matrix recovery via the RIP (to be
defined in Section \ref{sec:RIP}).  Subsequently, several papers
specialized the theory of nuclear-norm minimization to the matrix
completion problem \cite{CR08, CT09, Ma09, CP09, GrossGeneral09} which
turns out to be `RIPless'; this literature is motivated by very clear
applications such as recommender systems and network localizations,
and has required very sophisticated mathematical techniques.

With the recent increase in attention given to the low-rank matrix
model, which the authors surmise is due to the spring of new theory,
new applications are being quickly discovered that deviate from the
matrix completion setup (such as quantum state tomography
\cite{Gross09}), and could benefit from a different analysis.  Our
paper returns to measurement ensembles obeying the RIP as in
\cite{Recht07}, which are of a different nature than those involved in
matrix completion.  As in compressive sensing, the only known
measurement ensembles which provably satisfy the RIP at a nearly
minimal sampling rate are random (such as the Gaussian measurement
ensemble in Section \ref{sec:RIP}) Having said this, two comments are
in order.  First, our results provide an absolute benchmark of what is
achievable, thus allowing direct comparisons with other methods and
other sampling operators $\cA$. For instance, one can quantify how far
the error bounds for the RIPless matrix completion are from what is
then known to be essentially unimprovable. Second, since our results
imply that the restricted isometry property {\em alone} guarantees a
near-optimal accuracy, we hope that this will encourage more
applications with random ensembles, and also encourage researchers to establish
whether or not their measurements obey this desirable
property. Finally, we hope that our analysis offers insights for
applications with nonrandom measurement ensembles.

\subsection{Problem setup}
We observe data $y$
from the model
\begin{equation}
  \label{eq:model}
  y = {\cal A}(M) + z, 
\end{equation}
where $M$ is an unknown $n_1 \times n_2$ matrix, ${\cal A} : \R^{n_1
  \times n_2} \goto \R^m$ is a linear mapping, and $z$ is an
$m$-dimensional noise term.  The synthesized versions of our error
bounds assume that $z$ is a Gaussian vector with i.i.d.~${\cal
  N}(0,\sigma^2)$ entries.  The goal is to recover a good
approximation of $M$ while requiring as few measurements as possible.

We pause to demonstrate the form of $\cA(X)$ explicitly: the $i$th
entry of $\cA(X)$ is $[\cA(X)]_i = \<A_i, X\>$ for some sequence of
matrices $\{A_i\}$ and with the standard inner product $\<A,X\>
= \trace(A^* X)$.  Each $A_i$ can be likened to a row of a
compressive sensing matrix, and in fact it can aid the intuition to
think of $\cA$ as a large matrix, i.e.~one could write $\cA(X)$ as
\begin{equation}
\label{eq:bigA}
\cA(X) = \begin{bmatrix}
  \text{vec}(A_1)\\
  \text{vec}(A_2)\\
  \vdots\\
  \text{vec}(A_m)\\
\end{bmatrix}
\, \text{vec}(X), 
\end{equation}
where vec$(X)$ is a long vector obtained by stacking the columns of
$X$. In the common matrix completion problem, each $A_i$ is of the
form $e_k e_j^*$ so that the $i$th component of $\cA(X)$ is of the
form $\<e_k e_j^*, M\> = e_k^* M e_j = M_{kj}$ for some $(j,k)$.

\subsection{Algorithms}
To recover $M$, we propose solving one of two
nuclear-norm-minimization based algorithms.  The first is an analogue
to the Dantzig Selector from compressive sensing \cite{DS}, defined as
follows:
\begin{equation}
  \label{eq:ds}
   \begin{array}{ll}
     \textrm{minimize}   & \quad \|X\|_{*}\\ 
     \textrm{subject to} & \quad \|{\cal A}^* (r)\| \le \lambda\\
     & \quad r = y - {\cal A}(X), 
  \end{array}
\end{equation}
where the optimal solution is our estimate $\hat{M}$, $\opnorm{\cdot}$
is the operator norm and $\nucnorm{\cdot}$ is its dual, i.e.~the
nuclear norm, and $\cA^*$ is the adjoint of $\cA$. We call this convex
program the \textit{matrix Dantzig selector}.

To pick a useful value for the parameter $\lambda$ in \eqref{eq:ds},
we stipulate that the `true' matrix $M$ should be feasible (this is a
necessary condition for our proofs). In other words, one should have
$\opnorm{\cA^*(z)} \leq \lambda$; Section \ref{sec:minimax} provides
further intuition about this requirement.  In the case of Gaussian
noise, this corresponds to $\lambda = C n \sigma$ for some numerical
constant $C$ as in the following lemma.
\begin{lemma}
\label{lem:opNormz}
Suppose $z$ is a Gaussian vector with i.i.d.~$\mathcal{N}(0, \sigma^2)$ entries
and let $n = \max(n_1,n_2)$. Then if $C_0 > 4\sqrt{(1+\delta_1)\log
  12}$
\begin{equation}
  \label{eq:Atz}
  \|{\cal A}^*(z)\| \le C_0 \sqrt{n} \sigma, 
\end{equation}
with probability at least $1- 2e^{-cn}$ for a fixed numerical constant
$c > 0$.
\end{lemma}
This lemma is proved in Section \ref{sec:proofs} using a standard
covering argument. The scalar $\delta_1$ is the isometry constant at
$\rank 1$, as defined in Section \ref{sec:RIP}, but suffice for now
that it is a very small constant bounded by $\sqrt{2} - 1$ (with high probability)
under the assumptions of all of our theorems.

The optimization program \eqref{eq:ds} may be formulated as a
semidefinite program (SDP) and can thus be solved by any of the
standard SDP solvers. To see this, we first recall that the nuclear
norm admits an SDP characterization since $\|X\|_*$ is the optimal
value of the SDP
\[
  \begin{array}{ll}
    \text{minimize}   & \quad \bigl(\trace({W}_1) + \trace({W}_2)\bigr)/2   \\\\
 
    \text{subject to} & \quad \begin{bmatrix} {W}_1 & {X}\\
      {X}^* & {W}_2
 \end{bmatrix} \succeq 0
\end{array}
 \]
 with optimization variables $X, W_1, W_2 \in \R^{n\times n}$. Second,
 the constraint $\|{\cal A}^*(r)\| \le \lambda$ is an SDP constraint
 since it can be expressed as the linear matrix inequality (LMI)
\[
 \begin{bmatrix}  \lambda I_n & {\cal A}^*(r)\\
      [{\cal A}^*(r)]^* & \lambda I_n 
 \end{bmatrix} \succeq 0. 
\]
This shows that \eqref{eq:ds} can be formulated as the SDP
\[
  \begin{array}{ll}
    \text{minimize}   & \quad \bigl(\trace({W}_1) + \trace({W}_2)\bigr)/2 \vspace{2mm}\\
    \text{subject to} & \quad \begin{bmatrix} {W}_1 & {X} & 0 & 0\\
      {X}^* & {W}_2 & 0 & 0\\
      0 & 0 & \lambda I_n & {\cal A}^*(r)\\
      0 & 0 &  [{\cal A}^*(r)]^* & \lambda I_n 
 \end{bmatrix} \succeq 0 \vspace{2mm}\\
& \quad r = y - {\cal A}(X), 
\end{array}
\]
with optimization variables $X, W_1, W_2 \in \R^{n\times n}$.

However, a few algorithms have recently been developed to solve
similar nuclear-norm minimization problems without using
interior-point methods which work extremely efficiently in practice
\cite{SVT, Ma08}.  The nuclear-norm minimization problem solved using
fixed-point continuation in \cite{Ma08} is an analogue to the LASSO,
and is defined as follows:
\begin{equation}
\label{eq:lasso}
\text{minimize}  \quad \frac{1}{2} \twonorm{\cA(X) - y}^2 + \mu \nucnorm{X}.
\end{equation}
We call this convex program the \textit{matrix Lasso} and it is the
second convex program whose theoretical properties are analyzed in
this paper.

\subsection{Organization of the paper}
The results in this paper mostly concern random measurements and
random noise and so they hold with high probability.  In Section
\ref{sec:RIP}, we show that certain classes of random measurements
satisfy the RIP when only sampling a constant number of measurements
per degree of freedom.  In Section \ref{sec:minimax} we present the
simplest of our error bounds, demonstrating that when the RIP holds,
the solution to \eqref{eq:ds} is within a constant of the minimax
risk.  This error bound is refined in Section \ref{sec:oracle} to
provide a more adaptive error that holds improvements when the
singular values of $M$ decay below the noise level.  It is shown that
this error bound is within a constant of the expected value of a
certain `oracle' error bound.  In Section \ref{sec:instOpt}, we present
an error bound handling the case when $M$ has full rank but is well
approximated by a low-rank matrix. Section \ref{sec:proofs} contains
the proofs and we finish with some concluding remarks in Section
\ref{sec:conclusion}.

\subsection{Notation}
We review all notation used in this paper in order to ease
readability.  We assume $M \in \R^{n_1 \times n_2}$ and let $n =
\max(n_1, n_2)$. A variety of norms are used throughout this paper:
$\nucnorm{X}$ is the nuclear norm (the sum of the singular values);
$\opnorm{X}$ is the operator norm of $X$ (the top singular value);
$\fronorm{X}$ is the Frobenius norm (the $\ell_2$-norm of the vector
of singular values).  The matrix $X^*$ is the adjoint of $X$, and for
the linear operator $\cA :\R^{n_1 \times n_2} \rightarrow \R^m$,
$\cA^*: \R^m \rightarrow \R^{n_1 \times n_2}$ is the adjoint operator.
Specifically, if $[\cA(X)]_i = \<A_i, X\>$ for all matrices $X \in
\R^{n_1 \times n_2}$, then
\[\cA^*(q) = \sum_{i=1}^m q_i A_i\]
for all vectors $q \in \R^m$.


\section{Main Results}

\subsection{Matrix RIP}
\label{sec:RIP}
The matrix version of the RIP is an integral tool in proving our
theoretical results and we begin by defining the RIP in this setting
and describing measurement ensembles that satisfy it.  To characterize
the RIP, we introduce the isometry constants of a linear map $\cA$.
\begin{definition} For each integer $r = 1, 2, \ldots, n$, the
  isometry constant $\delta_r$ of ${\cal A}$ is the smallest quantity
  such that
  \begin{equation} \label{eq:rip}
  (1-\delta_r) \|X\|^2_F \le \twonorm{\cA(X)}^2 \le (1+\delta_r) \|X\|_F^2 
 \end{equation} 
holds for all matrices of rank at most $r$. 
\end{definition}
We say that $\cA$ satisfies the RIP at rank $r$ if $\delta_{r}$ is bounded by a sufficiently small constant between 0 and 1, the value of which will become apparent in further sections (see e.g. Theorem \ref{teo:DS1}).  

Which linear maps $\cA$ satisfy the RIP?  As a quintessential example,
we introduce the Gaussian measurement ensemble.
\begin{definition}
  $\cA$ is a Gaussian measurement ensemble if each `row' $A_i$, $1 \le
  i \le m$, contains i.i.d.~$\mathcal{N}(0,1/m)$ entries (and the $A_i$'s are
  independent from each other).
\end{definition}
This is of course highly analogous to the Gaussian random matrices in
compressive sensing.  Our first result is that Gaussian measurement
ensembles, along with many other random measurement ensembles, satisfy
the RIP when $m \geq C \, nr$ (with high probability) for some
constant $C > 0$.
\begin{theorem} \label{teo:RIP} Fix $0 \leq \delta < 1$ and let $\cA$
  be a random measurement ensemble obeying the following condition:
  for any given $X \in \R^{n_1\times n_2}$ and any fixed $0<t<1$,
\begin{equation}
\label{eq:concentration}
P(|\twonorm{\cA(X)}^2 - \fronorm{X}^2|  > t \fronorm{X}^2) \leq C \exp(-c m)
\end{equation}
for fixed constants $C,c>0$ (which may depend on $t$).  Then if $m
\geq D n r$, $\cA$ satisfies the RIP with isometry constant $\delta_r
\le \delta$ with probability exceeding $1 - C e^{-d m}$ for fixed
constants $D, d > 0$.
\end{theorem}

The many unspecified constants involved in the presentation of Theorem
\ref{teo:RIP} are meant to allow for general use with many random
measurement ensembles.  However, to make the presentation more
concrete we describe the constants involved in the concentration bound
\eqref{eq:concentration} for a few special random measurement
ensembles.  If $\cA$ is a Gaussian random measurement ensemble,
$\|\cA(X)\|_{\ell_2}^2$ is distributed as $m^{-1} \|X\|_F^2$ times a
chi-squared random variable with $m$ degrees of freedom and
\eqref{eq:concentration} follows from standard concentration
inequalities \cite{Laurent00, Recht07}.  Specifically, we have
\begin{equation}
\label{eq:setConcentration}
P(|\twonorm{\cA(X)}^2 - \fronorm{X}^2|  > t \fronorm{X}^2) \leq 2 \exp\left(-\frac{m}{2}(t^2/2 - t^3/3)\right).
\end{equation}
Similarly, $\cA$ satisfies equation \eqref{eq:setConcentration} in the
case when each entry of each `row' $A_i$ has i.i.d.~entries that are
equally likely to take the value $1/\sqrt{m}$ or $-1/\sqrt{m}$, or if
$\cA$ is a random projection \cite{Achlioptas, Recht07}.  Further,
$\cA$ satisfies \eqref{eq:concentration} if the `rows' $A_i$ contain
sub-Gaussian entries (properly normalized) \cite{Vershynin00},
although in this case the constants involved depend on the parameters
of the sub-Gaussian entries.

In order to ascertain the strength of Theorem \ref{teo:RIP}, note that
the number of degrees of freedom of an $n_1 \times n_2$ matrix of rank
$r$ is equal to $r(n_1 + n_2 - r)$.\footnote{This can be seen by
  counting the number of equations and unknowns in the singular value
  decomposition.}  Thus, one may expect that if $m < r(n_1 + n_2 -r)$,
there should be a rank-$r$ matrix in the null space of $\cA$ leading
to a failure to achieve the lower bound in \eqref{eq:rip}.  In order
to make this intuition rigorous (to within a constant) assume without
loss of generality that $n_2 \geq n_1$, and observe that the set of
rank-$r$ matrices contains all those matrices restricted to have
nonzero entries only in the first $r$ rows.  This is an $n\times r$
dimensional vector space and thus we must have $m \geq nr$ or
otherwise there will be a rank-$r$ matrix in the null space of $\cA$
regardless of what measurements are used. (This is a similar alternative to the null-space argument posed in the introduction.)

Theorem \ref{teo:RIP} is inspired by a similar theorem in
\cite{Recht07}[Theorem 4.2] and refines this result in two
ways. First, it shows that one only needs a constant number of
measurements per degree of freedom of the underlying rank-$r$ matrix
in order to obtain the RIP at rank $r$ (which improves on the result
in \cite{Recht07} by a factor of $\log n$ and also achieves the
theoretical lower bound to within a constant).  Second, it shows that
one must only require a single concentration bound on $\cA$, removing
another assumption required in \cite{Recht07}.  A possible third
benefit is that the proof follows simply and quickly from a
specialized covering argument. The novelty is in the method used to
cover low-rank matrices.

\subsection{The matrix Dantzig selector and the matrix Lasso are nearly minimax}
\label{sec:minimax}
In this section, we present our first and simplest error bound, which
only requires that $\cA$ satisfies the RIP.
\begin{theorem} 
\label{teo:DS1}
Assume that $\rank(M) \leq r$ and let $\hat{M}_{DS}$ be the solution
to the matrix Dantzig selector \eqref{eq:ds} and $\hat{M}_L$ be the
solution to the matrix Lasso \eqref{eq:lasso}.  If $\delta_{4r} <
\sqrt{2}-1$ and $\|{\cal A}^*(z)\| \le \lambda$ then
  \begin{equation}
    \label{eq:ds1}
    \|\hat{M}_{DS} - M\|^2_F \le C_0 \, r \lambda^2, 
  \end{equation}
  and if $\delta_{4r} < (3\sqrt{2} - 1)/17$ and $\opnorm{\cA^*(z)} \le
  \mu/2$, then
  \begin{equation}
  	\label{eq:lasso1}
  	\fronorm{\hat{M}_L - M}^2 \leq C_1 \, r \mu^2;
  \end{equation}
  above, $C_0$ and $C_1$ are small constants depending only on the
  isometry constant $\delta_{4r}$. In particular, if $z$ is a Gaussian
  error and $\hat M$ is either $\hat M_{DS}$ with $\lambda = 8n\sigma$, 
  or $\hat M_{L}$ with $\mu = 16n\sigma$, we have
 \begin{equation}
   \label{eq:dslasso}
    \|\hat M - M\|^2_F \le C'_0 \, n r \sigma^2
 \end{equation}
 with probability at least $1 - 2e^{-cn}$ for a constant $C'_0$
 (depending only on $\delta_{4r}$).
\end{theorem}
Note that \eqref{eq:dslasso} follows from \eqref{eq:ds1} and
\eqref{eq:lasso1} simply by plugging in $\lambda, \mu/2 = 8n\sigma$
into Lemma \ref{lem:opNormz}. In a nutshell, the error is proportional
to the number of degrees of freedom times the noise level.

An important point is that one may expect the error to be reduced when
further measurements are taken i.e.~one may expect the error to be
inversely proportional to $m$.  In fact, this is the case for the
Gaussian measurement ensemble, but this extra factor is absorbed into
the definition in order to normalize the measurements so that they
satisfy the RIP.  If instead, each row `$A_i$' in the Gaussian
measurement ensemble is defined to have i.i.d.~standard normal
entries, then by a simple rescaling argument (apply Theorem
\ref{teo:DS1} to $y/\sqrt{m}$), the error bound reads
\[\|\hat M - M\|^2_F \le C'_0 \, n r \sigma^2/m.\] A second remark is
that exploiting the low-rank structure helps to denoise.  For example,
if we measured every entry of $M$ (a measurement ensemble with
isometry constant $\delta_r = 0$), but with each measurement corrupted
by a $\mathcal{N}(0, \sigma^2)$ noise term, then taking the
measurements as they are as the estimate of $M$ would lead to an
expected error equal to
\[
\E \fronorm{\hat{M} - M}^2 = n^2 \sigma^2. 
\]
Nuclear-norm minimization\footnote{Of course if one sees all of the
  entries of the matrix plus noise, nuclear-norm minimization is
  unnecessary, and one can achieve minimax error bounds by truncating
  the singular values.} reduces this error by a factor of about $n/r$.

The strength of Theorem \ref{teo:DS1} is that the error bound
\eqref{eq:dslasso} is nearly optimal in the sense that no estimator
can do essentially better without further assumptions, as seen by
lower-bounding the expected minimax error.
\begin{theorem}
\label{teo:minimax}
  If $z$ is a Gaussian error, then any estimator $\hat M(y)$ obeys
  \begin{equation}
    \label{eq:minimax}
    \sup_{M : \rank(M) \le r} \, \E \|\hat M(y) - M\|_F^2 \ge \frac{1}{1+\delta_r} \, n r \sigma^2.    
  \end{equation}
  In other words, the minimax error over the class of matrices of rank
  at most $r$ is lower bounded by about $n r \sigma^2$.
\end{theorem}

Before continuing, it may be helpful to analyze the solutions to the
matrix Dantzig selector and the matrix Lasso in a simple case in order
to understand the error bounds in Theorem \ref{teo:DS1}
intuitively, and also to understand our choice of $\lambda$ and $\mu$.
Suppose $\cA$ is the identity so that changing the notation a bit, the
model is $Y = M + Z$, where $Z$ is an $n \times n$ matrix with
i.i.d.~Gaussian entries.  We would like the unknown matrix $M$ to be a
feasible point, which requires that $\opnorm{Z} \leq \lambda$ (for example,
if $\opnorm{Z} > \lambda$, we already have problems when $M = 0$).  It
is well known that the top singular value of a square $n \times n$
Gaussian matrix, with per-entry variance $\sigma^2$, is concentrated
around $\sqrt{2n} \sigma$, and thus we require $\lambda \geq \sqrt{2n}
\sigma$ (this provides a slightly sharper bound than Lemma
\ref{lem:opNormz}). Let $T_\lambda(X)$ denote the singular value
thresholding operator given by 
\[
T_\lambda(X) = \sum_i \max(\sigma_i(X) - \lambda,0) u_i v_i^*, 
\]
where $X = \sum_i \sigma_i(X) u_i v_i^*$ is any singular value
decomposition. In this simple setting, the solution to \eqref{eq:ds}
and \eqref{eq:lasso} can be explicitly calculated, and for $\lambda =
\mu$ they are both equal to $T_\lambda(M + Z)$.  If $\lambda$ is too
large, then $T_\lambda(M + Z)$ becomes strongly biased towards zero,
and thus (loosely) $\lambda$ should be as small as possible while
still allowing $M$ to be feasible for the matrix Dantzig selector
\eqref{eq:ds}, leading to the choice $\lambda \approx \sqrt{2n}
\sigma$.

Further, in this simple case we can calculate the error bound in a few
lines. We have
\begin{align*}
\opnorm{\hat{M} - M} &= \opnorm{T_\lambda(Y) - Y + Z}\\
&\leq \opnorm{T_\lambda(Y) - Y} + \opnorm{Z}\\
&\leq 2\lambda
\end{align*}
assuming that $\lambda \geq \opnorm{Z}$.  Then 
\begin{align}
  \nonumber \fronorm{\hat{M} - M}^2 & \leq \opnorm{\hat{M} - M}^2\, 
  \rank(\hat{M} - M) \\
  & \leq 4 \lambda^2 \, \rank(\hat{M} - M).
\end{align}
Once again, assuming that $\lambda \geq \opnorm{Z}$, we have
$\rank(\hat{M} - M) \leq \rank(\hat{M}) + \rank(M) \leq 2r$.  Plugging
this in with $\lambda = C \sqrt{n} \sigma$ gives the error bound
\eqref{eq:dslasso}.

\subsection{Oracle inequalities}
\label{sec:oracle}

Showing that an estimator achieves the minimax risk is reassuring but
is sometimes not considered completely satisfactory. As is
frequently discussed in the literature, the minimax approach focuses
on the worst-case performance and it is quite reasonable to expect
that for matrices of general interest, better performances are
possible. In fact, a recent trend in statistical estimation is to
compare the performance of an estimator with what is achievable with
the help of an oracle that reveals extra information about the problem. A good match indicates an overall excellent
performance.

To develop an oracle bound, assume w.l.o.g.~that $n_2 \geq n_1$ so
that $n = n_2$, and consider the family of estimators defined as
follows: for each $n_1 \times r$ orthogonal matrix $U$, define
\begin{equation}
  \label{eq:MU}
  \hat{M}[U] = \arg \min \{\|y - {\cal A}(\hat M)\|_{\ell_2} : 
\hat M = UR \text{ for some } R\}.   
\end{equation}
In other words, we fix the column space (the linear space spanned by
the columns of the matrix $U$), and then find the matrix with that
column space which best fits the data. Knowing the true matrix $M$, an
oracle or a genie would then select the best column space to use as to
minimize the mean-squared error (MSE)
\begin{equation}
  \label{eq:oracle}
 \inf_U \E \|M - \hat{M}[U]\|^2. 
\end{equation}
The question is whether it is possible to mimic the performance of the
oracle and achieve a MSE close to \eqref{eq:oracle} with a real
estimator.

Before giving a precise answer to this question, it is useful to
determine how large the oracle risk is. To this end, consider a fixed
orthogonal matrix $U$, and write the least-squares estimate
\eqref{eq:MU} as
\[
\hat{M}[U] := U {\cal H}_U(y), \quad {\cal H}_U = ({\cal A}_U^*{\cal
  A}_U)^{-1} {\cal A}_U^*,
\]
where ${\cal A}_U$ is the linear map
\begin{equation}
\label{eq:AU}
\begin{array}{lllcl}
{\cal A}_U & : & \R^{r\times n} & \goto & \R^{m}\\
& &  R & \mapsto & {\cal A}(UR), 
\end{array}
\end{equation}
and 
\[
\begin{array}{lllcl}
  {\cal A}_U^* & : & \R^{m} & \goto & \R^{r \times n}\\
  & & y & \mapsto & U^*{\cal A}^*(y).
\end{array}
\]
Then decompose the MSE as the sum of the squared bias and variance 
\begin{align*}
\E \|M - \hat{M}[U]\|_F^2 & = \|\text{bias}\|^2 + \text{variance}\\
& = \|\E  \hat{M}[U] - M\|_F^2 + \E \|U {\cal H}_U(z)\|_F^2.
\end{align*}
The variance term is classically equal to
\[
\E \fronorm{U {\cal H}_U(z)}^2 = \E \|{\cal H}_U(z)\|_F^2 = \sigma^2 \trace({\cal H}_U^*{\cal H}_U) =
\sigma^2 \trace( ({\cal A}_U^*{\cal A}_U)^{-1}).
\]
Due to the restricted isometry property, all the eigenvalues of the
linear operator ${\cal A}_U^*{\cal A}_U$ belong to the interval
$[1-\delta_r, 1+\delta_r]$, see Lemma \ref{teo:AU}. Therefore, the
variance term obeys
\[
\sigma^2 \trace( ({\cal A}_U^*{\cal A}_U)^{-1}) \ge
\frac{1}{1+\delta_r}\, nr \sigma^2.
\]
For the bias term, we have 
\[
\E \hat{M}[U] - M = U ({\cal A}_U^*{\cal A}_U)^{-1} {\cal A}_U^* {\cal
  A}(M) - M, 
\]
which we rewrite as 
\begin{align*}
\E  \hat{M}[U] - M & = U  ({\cal A}_U^*{\cal
  A}_U)^{-1} {\cal A}_U^* {\cal A}((I - UU^* + UU^*)M) - M\\
& = U  ({\cal A}_U^*{\cal
  A}_U)^{-1} {\cal A}_U^* {\cal A}((I - UU^*)M) +  U  ({\cal A}_U^*{\cal
  A}_U)^{-1} {\cal A}_U^* {\cal A}_U(U^*M) - M\\
& =   U  ({\cal A}_U^*{\cal
  A}_U)^{-1} {\cal A}_U^* {\cal A}((I - UU^*)M) - (I - UU^*)M. 
\end{align*}
Hence, the bias is the sum of two matrices: the first has a column
space included in the span of the columns of $U$ while the column
space of the other is orthogonal to this span. Put $P_{U^\perp}(M) =
(I - UU^*)M$; that is, $P_{U^\perp}(M)$ is the (left) multiplication
with the orthogonal projection matrix $(I - UU^*)$. We have
\begin{align*}
\|\E \hat{M}[U] - M\|^2 & = \|U ({\cal A}_U^*{\cal A}_U)^{-1} {\cal
  A}_U^* {\cal A}(P_{U^\perp}(M))\|_F^2 + \|P_{U^\perp}(M)\|_F^2\\
& \ge \|P_{U^\perp}(M)\|_F^2.
\end{align*}

 To summarize, the oracle bound obeys
\[
\inf_U \E \|M - \hat{M}[U]\|^2 \ge \inf_U \, \left[
\|P_{U^\perp}(M)\|_2^2 + \frac{nr \sigma^2}{1+\delta_r} \right]. 
\]
Now for a given dimension $r$, the best $U$ -- that minimizing the
squared bias term or its proxy $\|P_{U^\perp}(M)\|_F^2$ -- spans the
top $r$ singular vectors of the matrix $M$. Denoting the singular
values of $M$ by $\sigma_i(M)$, we obtain
\[
\inf_U \E \|M - \hat{M}[U]\|^2 \ge \inf_r \, \left[\sum_{i > r}
  \sigma_i^2(M)
  \ + \frac{1}{2} nr \sigma^2\right],
\]
which for convenience we simplify to
\begin{equation}
  \label{eq:oraclebound}
  \inf_U \E \|M - \hat{M}[U]\|^2 \ge \frac{1}{2} 
\sum_i \min(\sigma_i^2(M), n\sigma^2). 
\end{equation}
The right-hand side has a nice interpretation. Write the SVD of $M$ as
$ M = \sum_{i = 1}^r \sigma_i(M) u_i v_i^*$.  Then if $\sigma_i^2(M) >
n \sigma^2$, one should try to estimate the rank-$1$ contribution
$\sigma_i(M)\, u_i v_i^*$ and pay the variance term (which is about
$n\sigma^2$) whereas if $\sigma_i^2(M) \le n \sigma^2$, we should not
try to estimate this component, and pay a squared bias term equal to
$\sigma_i^2(M)$. In other words, the right-hand side may be
interpreted as an {\em ideal} bias-variance trade-off. 

The main result of this section is that the matrix Dantzig Selector
and matrix Lasso achieve this same ideal bias-variance trade-off to
within a constant.
\begin{theorem} 
\label{teo:oracle}
Assume that $\rank(M) \leq r$ and let $\hat{M}_{DS}$ be the solution
to the matrix Dantzig selector \eqref{eq:ds} and $\hat{M}_L$ be the
solution to the matrix Lasso \eqref{eq:lasso}.  Suppose $z$ is a Gaussian error and let
$\lambda = 16n\sigma$ and $\mu = 32n\sigma^2$. If $\delta_{4r} <
\sqrt{2} - 1$, then
  \begin{equation}
    \label{eq:ds21}
    \|\hat{M}_{DS} - M\|^2_F \le C_0 \, \sum_i \min(\sigma_i^2(M), n\sigma^2),  
\end{equation}
and if $\delta_{4r} < (3\sqrt{2} - 1)/17$, then
\begin{equation}
	\label{eq:lasso21}
	\fronorm{\hat{M}_L - M}^2 \leq C_1 \, \sum_i \min(\sigma_i^2(M), n\sigma^2) 
\end{equation}
with probability at least $1 - 2e^{-cn}$ for constants $C_0$ and $C_1$
(depend only on $\delta_{4r}$).
\end{theorem}
In other words, not only does nuclear-norm minimization mimic the performance that one would achieve with an oracle that gives the exact column space of $M$ (as in Theorem \ref{teo:minimax}), but in fact the error bound is within a constant of what one would achieve by projecting onto the optimal column space corresponding only to the significant singular values.

While a similar result holds in the compressive sensing literature
\cite{DS}, we derive the result here using a novel technique. We use a
middle estimate $\bar{M}$ which is the optimal solution to a certain
rank-minimization problem (see Section \ref{sec:proofs}) and is
provably near  $\hat{M}$ and $M$.  With this technique, the proof is
a fairly simple extension of Theorem \ref{teo:DS1}.

\subsection{Extension to full-rank matrices}
\label{sec:instOpt}
In some applications, such as sensor localization, $M$ has exactly low
rank, i.e.~only the top few of its singular values are nonzero.
However, in many applications, such as quantum state tomography, $M$
has full rank, but is well approximated by a low-rank matrix.  In this
section, we demonstrate an extension of the preceding error bound when
$M$ has full rank.

First, suppose $n_1 \leq n_2$ and note that a result of the form
\begin{equation}
\label{eq:fullBiasVar}
\fronorm{\hat{M} - M}^2  \leq C \sum_{i=1}^{n_1} \min(\sigma_i^2(M), n\sigma^2)
\end{equation}
would be impossible when undersampling $M$ because it would imply that as the noise level $\sigma$ approaches zero, an arbitrary full-rank $n\times n$ matrix could be exactly reconstructed from fewer than $n^2$ linear measurements.  Instead, our result essentially splits $M$ into two parts, 
\[
M = \sum_{i=1}^{\bar{r}} \sigma_i(M) u_i v_i^* +
\sum_{i=\bar{r}+1}^{n_1} \sigma_i(M) u_i v_i^* = M_{\bar{r}} + M_c
\]
where $\bar{r} \approx m/n$, and $M_{\bar{r}}$ is the best
rank-$\bar{r}$ approximation to $M$.  The error bound in the theorem
reflects a near-optimal bias-variance trade-off in recovering
$M_{\bar{r}}$, but an inability to recover $M_c$ (and indeed the proof
essentially considers $M_c$ as non-Gaussian noise).  Note that $\bar r
(n_1 + n_2 - \bar r)$ is of the same order as $m$ so that the part of
the matrix which is well recovered has about as many degrees of
freedom as the number of measurements.  
In other words, even in the noiseless case this theorem demonstrates instance optimality i.e. the error bound is proportional to the norm of the part of $M$ that is irrecoverable given the number of measurements (see \cite{wojtaszczyk08} for an analogous result in compressive sensing).  In the noisy case there does not seem to be any current analogue to this error bound in compressive sensing, although the detailed analysis can be translated to the compressive sensing problem and the authors are currently writing a short paper containing this result. 

\begin{theorem} \label{teo:instanceOpt} Fix $M$.  Suppose that $\cA$
  is sampled from the Gaussian measurement ensemble with $m \leq c_0
  n^2/\log(m/n)$ and let $\bar{r} \leq c_1 m/n$ for some fixed
  numerical constants $c_0$ and $c_1$. Let $\hat{M}$ be the solution
  to the matrix Dantzig selector $\eqref{eq:ds}$ with $\lambda = 16
  \sqrt{n} \sigma$ or the solution to the matrix Lasso
  \eqref{eq:lasso} with $\mu = 32 \sqrt{n} \sigma$. Then
\begin{equation}
\label{eq:instOpt}
\fronorm{\hat{M} - M}^2 \leq C\left(\sum_{i=1}^{\bar{r}} \min(\sigma_i^2(M), n\sigma^2) + \sum_{i = \bar{r}+1}^n \sigma_i^2(M) \right)
\end{equation}
with probability greater than $1 - De^{-dn}$ for fixed numerical
constants $C,D, d> 0$. Roughly, the same conclusion extends to
operators obeying the NNQ condition, see below.
\end{theorem}

An interesting note is that in the noiseless case this error bound provides a case of `instance optimality'

First note that $\bar{r}$ is small enough so that the RIP holds with high
probability (see Lemma \ref{teo:RIP}).  However, the theorem requires
more than just the RIP.  The other main requirement is a certain NNQ
condition, which holds for Gaussian measurement ensembles and is
introduced in Section \ref{sec:proofs}.  It is an analogous
requirement to the LQ condition introduced by Wojtaszczyk
\cite{wojtaszczyk08} in compressive sensing. To keep the presentation
of the Theorem simple, we defer the explanation of the NNQ condition
to the proofs section and simply state the theorem for the Gaussian
measurement ensemble.  However, the proof is not sensitive to the use
of this ensemble (for example sub-Gaussian measurements yield the same
result). Many generalizations of this Theorem are available and the
lemmas necessary to make such generalizations are spelled out in
Section \ref{sec:proofs}.

The assumption that $m \leq c n^2/\log(m/n)$ seems to be an artifact
of the proof technique.  Indeed, one would not expect further
measurements to negatively impact performance.  In fact, when $m \geq
c' n^2$ for a fixed constant $c'$, one can use Lemma
\ref{teo:nucBound} from Section \ref{sec:proofs} to derive the error
bound \eqref{eq:instOpt} (with high probability), leaving the
necessity for a small `patch' in the theory when $cn^2/\log(m/n) \leq
n \leq c'n^2$.  However, our results intend to address the situation
in which $M$ is significantly undersampled, i.e. $m \ll n^2$, so the
requirement that $m \leq c n^2/\log(m/n)$ should be intrinsic to the
problem setup.

\section{Proofs}
\label{sec:proofs}

The proofs of several of the theorems use $\epsilon$-nets.  For a set
$S$, an $\epsilon$-net $S_{\epsilon}$ with respect to a norm
$\opnorm{\cdot}$ satisfies the following property: for any $v \in
S$, there exists $v_0 \in S_{\epsilon}$ with $\opnorm{v_0 - v} \leq
\epsilon$.  In other words, $S_\epsilon$ approximates $S$ to within
distance $\epsilon$ with respect to the norm $\opnorm{\cdot}$.  As shown in \cite{VershyninNotes}, there always exists an $\epsilon$-net $S_\epsilon$ satisfying $S_\epsilon \subset S$ and
\[|S_\epsilon| \leq \frac{\text{Vol} \left(S + \frac{1}{2} D\right)}{\text{Vol}\left(\frac{1}{2} D\right)}\]
where $\frac{1}{2} D$ is an $\epsilon/2$ ball (with respect to the norm $\opnorm{\cdot}$) and $S + \frac{1}{2} D = \{x+y: x \in S, y \in \frac{1}{2} D\}$.  In particular, if $S$ is a unit ball in $n$ dimensions (with respect to the norm $\opnorm{\cdot}$) or if it is the surface of the unit ball or any other subset of the unit ball, then $S + \frac{1}{2}D$ is contained in the $1+ \epsilon/2$ ball, and the thus
\[|S_\epsilon| \leq \frac{(1 + \epsilon/2)^n}{(\epsilon/2)^n} = \left(\frac{2 + \epsilon}{\epsilon}\right)^n \leq \left(3/\epsilon\right)^n\]
where the last inequality follows because we always take $\epsilon \leq 1$.  See \cite{VershyninNotes} for a more detailed argument.  We will require in all of our proofs that $S_\epsilon \subset S$.


\subsection{Proof of Lemma \ref{lem:opNormz}}

We assume that $\sigma = 1$ without loss of generality.  Put $Z =
{\cal A}^*(z)$. The norm of $Z$ is given by
\[
\|Z\| = \sup \,\, \<w, Z v\>,
\]
where the supremum is taken over all pairs of vectors on the unit
sphere $S^{n-1}$. Consider a $1/4$-net ${\cal N}_{1/4}$ of $S^{n-1}$
with $|{\cal N}_{1/4}| \leq 12^n$. For each $v,w \in S^{n-1}$,
\begin{align*}
  \<w, Z v\> & = \<w - w_0, Z v\> + \<w_0, Z (v-v_0)\> + \<w_0, Z v_0\>\\
  & \le \|Z\| \|w - w_0\|_{\ell_2} + \|Z\| \|v - v_0\|_{\ell_2} + \<w_0, Z v_0\>
\end{align*}
for some $v_0, w_0 \in {\cal N}_{1/4}$ obeying $\|v - v_0\|_{\ell_2} \le 1/4$,
$\|w - w_0\|_{\ell_2} \le 1/4$. Hence, 
\[
\|Z\| \le 2 \sup_{v_0, w_0 \in {\cal N}_{1/4}} \,\, \<w_0, Z v_0\>.   
\] 
Now for a fixed pair $(v_0,w_0)$, 
\[
\<w_0, Z v_0\> = \trace(w_0^* {\cal A}^*(z) v_0) = \trace(v_0 w_0^*
{\cal A}^*(z)) = \<w_0 v_0^*, {\cal A}^*(z)\> = \<{\cal A}(w_0 v_0^*),
z\>. 
\]
We deduce from this that $\<w_0, Z v_0\> \sim {\cal N}(0, \|{\cal
  A}(w_0 v_0^*)\|_{\ell_2}^2)$. Now
\[
 \|{\cal
  A}(w_0 v_0^*)\|_{\ell_2}^2 \le (1+\delta_1) \|w_0 v_0^*\|_F^2 = (1+\delta_1)
\]
so that by a standard tail bound for Gaussian random variables
\[
\P(|\<w_0, Z v_0\>| \ge \lambda) \le 2
e^{-\frac{1}{2}\frac{\lambda^2}{1+\delta_1}}.
\]
Therefore, 
\[
\P(\max |\<w_0, Z v_0\>| \ge \gamma \sqrt{(1+\delta_1) n}) \le 2
|{\cal N}_{1/4}|^2 e^{-\frac12 \gamma^2 n} \le 2 e^{2n\log
  12 - \frac12 \gamma^2 n},
\]
which is bounded by $2e^{-cn}$ with $c = \gamma^2/2 - 2 \log 12$ (we require $\gamma > 2 \sqrt{\log 12}$ so that $c > 0$).

\subsection{Proof of Theorem \ref{teo:RIP}}
The proof uses a covering argument, starting with the following lemma.
\begin{lemma}[Covering number for low-rank matrices] 
\label{teo:covering}
Let $S_r = \{X \in \R^{n_1 \times n_2} : \rank(X) \leq r, \,
\fronorm{X} = 1\}$.  Then there exists an $\epsilon$-net $\bar{S}_r$
for the Frobenius norm obeying 
\[
|\bar{S}_r| \leq (9/\epsilon)^{(n_1+n_2+1)r}.
\]
\end{lemma}

\begin{proof}
  Recall the SVD $X = U \Sigma V^*$ of any $X \in S_r$ obeying
  $\fronorm{\Sigma} = 1$. Our argument constructs an $\epsilon$-net for
  $S_r$ by covering the set of permissible $U,V$ and $\Sigma$. We work
  in the simpler case where $n_1 = n_2 = n$ since the general case is
  a straightforward modification.

  Let $D$ be the set of diagonal matrices with nonnegative diagonal
  entries and Frobenius norm equal to one. We take $\bar{D}$ to be an
  $\epsilon/3$-net for $D$ with $|\bar{D}| \leq (9/\epsilon)^r$.
  Next, let $O_{n,r} = \{U \in \R^{n \times r} : U^* U = I\}$. To
  cover $O_{n,r}$, it is beneficial to use the $\onetwonorm{\cdot}$
  norm defined as
\[
\onetwonorm{X} = \max_i \twonorm{X_i}, 
\]
where $X_i$ denotes the $i$th column of $X$. Let $Q_{n,r} = \{X \in
\R^{n \times r} : \onetwonorm{X} \leq 1\}$. It is easy to see that
$O_{n,r} \subset Q_{n,r}$ since the columns of an orthogonal matrix
are unit normed.  We have seen that there is an $\epsilon/3$-net
$\bar{O}_{n,r}$ for $O_{n,r}$ obeying $|\bar{O}_{n,r}| \leq
(9/\epsilon)^{n r}$.  We now let $\bar{S}_r =
\{\bar{U}\bar{\Sigma}\bar{V}^* : \bar{U}, \bar V \in O_{n,r}, \,
\bar{\Sigma} \in \bar{D}\}$, and remark  that $|\bar{S}_r| \le
|\bar{O}_{n,r}|^2 \, |\bar{D}| \leq (9/\epsilon)^{(2n+1)r}$.  
It remains to show that for all $X \in S_r$ there exists $\bar{X} \in
\bar{S}_r$ with $\fronorm{X - \bar{X}} \leq \epsilon$.

Fix $X \in S_r$ and decompose $X$ as $X = U\Sigma V^*$ as above.  Then
there exist $\bar{X} = \bar{U}\bar{\Sigma}\bar{V}^* \in \bar{S}_{r}$
with $\bar{U}, \bar V \in \bar{O}_{n,r}$, $\bar{\Sigma} \in \bar{D}$
obeying $\onetwonorm{U - \bar{U}} \leq \epsilon/3, \onetwonorm{V -
  \bar{V}} \leq \epsilon/3$, and $\fronorm{\Sigma - \bar{\Sigma}} \leq
\epsilon/3$.  This gives
\begin{align}
  \fronorm{X - \bar{X}} &= \fronorm{U\Sigma V^* - \bar{U} \bar{\Sigma} \bar{V}^*} \nonumber\\
  &= \fronorm{U\Sigma V^* - \bar{U} \Sigma V^* + \bar{U} \Sigma V^* - \bar{U} \bar{\Sigma} V^* + \bar{U} \bar{\Sigma} V^* - \bar{U} \bar{\Sigma} \bar{V}^*} \nonumber\\
  &\leq \fronorm{(U - \bar{U})\Sigma V^*} + \fronorm{\bar{U} (\Sigma -
    \bar{\Sigma})V^*} + \fronorm{\bar{U} \bar{\Sigma} (V -
    \bar{V})^*}.
\label{eq:coveringBound}
\end{align}
For the first term, note that since $V$ is an orthogonal matrix,
$\fronorm{(U - \bar{U})\Sigma V^*} = \fronorm{(U - \bar{U})\Sigma}$,
and
\begin{align*}
  \fronorm{(U - \bar{U})\Sigma}^2 &= \sum_{1\leq i \leq r} \Sigma_{i,i}^2 \twonorm{\bar{U}_i - U_i}^2\\
  &\leq \fronorm{\Sigma}^2 \onetwonorm{U - \bar{U}}^2\\
  &\leq (\epsilon/3)^2. 
\end{align*}
Hence, $\fronorm{(U - \bar{U})\Sigma V^*} \leq \epsilon/3$.  The same
argument gives $\fronorm{\bar{U} \bar{\Sigma} (V - \bar{V})^*} \leq
\epsilon/3$.  To bound the middle term, observe that $\fronorm{\bar{U}
  (\Sigma - \bar{\Sigma})V^*} = \fronorm{\Sigma - \bar{\Sigma}} \leq
\epsilon/3$.
This completes the proof. 
\end{proof}

We now prove Theorem \ref{teo:RIP}. It is a standard argument from
this point, and is essentially the same as the proof of Lemma 4.3 in
\cite{Recht07}, but we repeat it here to keep the paper
self-contained. We begin by showing that $\mathcal{A}$ is an
approximate isometry on the covering set $\bar{S}_r$.  Lemma
\ref{teo:covering} with $\epsilon = \delta/(4\sqrt{2})$ gives
\begin{equation}
|\bar{S}_r| \leq (36\sqrt{2}/\delta)^{(n_1+n_2+1)r}.
\end{equation}
Then it follows from \eqref{eq:concentration} together with the union
bound that
\begin{align}
\P\left(\max_{\bar{X} \in \bar{S}_r} |\twonorm{\mathcal{A}(\bar{X})}^2 -\fronorm{\bar{X}}^2| > \delta/2\right)  &\leq  |\bar{S}_r| C e^{-cm} \nonumber\\
&\leq 2(36\sqrt{2}/\delta)^{(n_1+n_2+1)r}C e^{-cm} \nonumber\\
&= C\exp\left((n_1+n_2+1)r \log(36\sqrt{2}/\delta) - cm\right) \nonumber\\
&\leq 2 \exp(-d m) \nonumber
\end{align}
where $d = c - \frac{\log(36\sqrt{2}/\delta)}{C}$ and we plugged in
both requirements $m \geq C (n_1+n_2+1)r$ and $C >
\log(36\sqrt{2}/\delta)/c$.

Now suppose that
\[\max_{\bar{X} \in \bar{S}_r} |\hspace{.5mm}\twonorm{\mathcal{A}(\bar{X})}^2 -\fronorm{\bar{X}}^2| \leq \delta/2\]
(which occurs with probability at least $1 - C\exp(-dm)$). 
We begin by showing that the upper bound in the RIP condition
holds. Set
\[
\kappa_r = \sup_{X \in S_r} \twonorm{\mathcal{A}(X)}.
\] 
For any $X \in S_r$, there exists $\bar{X} \in \bar{S}_r$ with
$\fronorm{X - \bar{X}} \leq \delta/(4\sqrt{2})$ and, therefore,
\begin{equation}
\label{eq:upperbound}
\twonorm{\mathcal{A}(X)} \leq \twonorm{\mathcal{A}(X-\bar{X})}+\twonorm{\mathcal{A}(\bar{X})} \leq \twonorm{\mathcal{A}(X - \bar{X})} + 1 + \delta/2.
\end{equation}
Put $\Delta X = X - \bar{X}$ and note that $\rank(\Delta X) \leq 2r$.
Write $\Delta X = \Delta X_1 + \Delta X_2$, where $\<\Delta X_1,
\Delta X_2\> = 0$, and $\rank(\Delta X_i) \leq r$, $i = 1,2$ (for
example by splitting the SVD).  Note that $\Delta X_1/\fronorm{\Delta
  X_1}$, $\Delta X_2/\fronorm{\Delta X_2} \in S_r$ and, thus, 
\begin{equation}
  \twonorm{\mathcal{A}(\Delta X)} \leq \twonorm{\mathcal{A}(\Delta X_1)} + \twonorm{\mathcal{A}(\Delta X_2)} \leq \kappa_r(\fronorm{\Delta X_1} + \fronorm{\Delta X_2}).
\end{equation}
Now $\fronorm{\Delta X_1} + \fronorm{\Delta X_2} \leq \sqrt{2}
\fronorm{\Delta X}$ which follows from $\fronorm{\Delta X_1}^2 +
\fronorm{\Delta X_2}^2 = \fronorm{\Delta X}^2$.  Also,
$\fronorm{\Delta X} \leq \delta/(4 \sqrt{2})$ leading to
$\twonorm{\mathcal{A}(\Delta X)} \leq \delta/4$.  Plugging this into
\eqref{eq:upperbound} gives
\[
\twonorm{\mathcal{A}(X)} \leq \kappa_r \delta/4 + 1 + \delta/2.
\]
Since this holds for all $X \in S_r$, we have $\kappa_r \leq \kappa_r
\delta/4 + 1 + \delta/2$ and thus $\kappa_r \leq (1 + \delta/2)/(1 -
\delta/4) \leq 1+\delta$ which essentially completes the upper bound.
Now that this is established, the lower bound now follows from
\[
\twonorm{\mathcal{A}(X)} \geq \twonorm{\mathcal{A}(\bar{X})} - \twonorm{\mathcal{A} \Delta X} \geq 1 - \delta/2 - (1+\delta)\sqrt{2}\delta/(4\sqrt{2}) \geq 1 - \delta. 
\]

Note that we have shown
\[
(1 - \delta) \fronorm{X} \leq \|\cA(X)\|_{\ell_2} \leq (1 + \delta)
\fronorm{X}, 
\]
which can then be easily translated into the desired version of the
RIP bound.

\subsection{Proof of Theorem \ref{teo:DS1}}
We prove Theorems \ref{teo:DS1}, \ref{teo:oracle}, and
\ref{teo:instanceOpt} for the matrix Dantzig selector \eqref{eq:ds}
and describe in Section \ref{sec:extension} how to extend these proofs
to the matrix Lasso.  We also assume that we are dealing with square
matrices from this point forward ($n = n_1 = n_2$) for notational
simplicity; the generalizations of the proofs to rectangular matrices
are straightforward.

We begin by a lemma, which applies to full-rank matrices, and contains
Theorem \ref{teo:DS1} as a special case.\footnote{We did not present
  this lemma in the main portion of the paper because it does not seem
  to have an intuitive interpretation.}
\begin{lemma} 
\label{teo:nucBound}
Suppose $\delta_{4r} < \sqrt{2}-1$ and let $M_r$ be any rank-r matrix.
Let $M_c = M - M_r$. Suppose $\lambda$ obeys $\|{\cal A}^*(z)\| \le
\lambda$. Then the solution $\hat M$ to \eqref{eq:ds} obeys
  \begin{equation}
    \|\hat M - M\|_F \le C_0 \, \sqrt{r} \lambda + C_1 \|M_c\|_*/r,  
  \end{equation}
  where $C_0$ and $C_1$ are small constants depending only on the isometry
  constant $\delta_{4r}$. 
\end{lemma}
We shall use the fact that $\cA$ maps low-rank orthogonal matrices to
approximately orthogonal vectors.
\begin{lemma}\cite{CT05}
  \label{teo:theta} For all $X$, $X'$ obeying $\<X,X'\> = 0$, and
  $\rank(X) \le r$, $\rank(X') \le r'$, 
\[
|\langle {\cal A}(X), {\cal A}(X')\rangle| \leq \delta_{r + r'} \,
\|X\|_F \, \|X'\|_F.
\]
\end{lemma}
\begin{proof} This is a simple application of the parallelogram
  identity. Suppose without loss of generality that $X$ and $X'$ have
  unit Frobenius norms. Then
  \[
  (1-\delta_{r+r'}) \|X \pm X'\|_F^2 \le \|{\cal A}(X \pm X')\|_F^2 \le
  (1+\delta_{r+r'}) \|X \pm X'\|_F^2, 
\]
since $\rank(X \pm X') \le r + r'$. We have $\|X \pm X'\|_F^2 = \|X\|_F^2
+ \|X'\|_F^2 = 2$ and the parallelogram identity asserts that
\[
|\langle {\cal A}(X), {\cal A}(X')\rangle|  = \frac{1}{4} \left|  \|{\cal A}(X + X')\|_F^2 -  \|{\cal A}(X - X')\|_F^2 \right| \le
\delta_{r + r'}, 
\]
which concludes the proof.
\end{proof}

The proof of Lemma \ref{teo:nucBound} parallels that of Cand\`es and Tao about
the recovery of nearly sparse vectors from a limited number of
measurements \cite{DS}. It is also inspired by the work of Fazel,
Recht, Cand\`es and Parrilo \cite{FCRP08, Recht07}.  Set $H = \hat M - M$ and
observe that by the triangle inequality,
\begin{equation}
  \label{eq:tube}
\|{\cal A}^*{\cal A}(H)\| \le \|{\cal A}^*({\cal A}(\hat M) - y)\| + \|{\cal A}^*(y - {\cal
  A}(M))\| \le 2\lambda, 
\end{equation}
since $M$ is feasible for the problem \eqref{eq:ds}. Decompose $H$
as
\[
H = H_0 + H_c,  
\]
where $\rank(H_0) \le 2r$, $M_r H_c^* = 0$ and $M_r^* H_c = 0$ (see \cite{Recht07}). We
have
\begin{align*}
  \|M + H\|_* & \ge \|M_r + H_c\|_* - \|M_c\|_* - \|H_0\|_* \\
  & = \|M_r\|_* + \|H_c\|_* - \|M_c\|_* - \|H_0\|_*.
\end{align*}
Since by definition, $\|M + H\|_* \le \|M\|_* \leq \|M_r\|_* + \|M_c\|_*$, this gives 
\begin{equation}
  \label{eq:cone}
  \|H_c\|_* \le \|H_0\|_* + 2 \|M_c\|_*.
\end{equation}

Next, we use a classical estimate developed in \cite{CRT2} (see also
\cite{Recht07}). Let $H_c = U \textrm{diag}(\vec{\sigma}) V^*$ be the
SVD of $H_c$, where $\vec{\sigma}$ is the list of ordered singular
values (not to be confused with the noise standard deviation).
Decompose $H_c$ into a sum of matrices $H_1, H_2, \ldots$, each of
rank at most $2r$ as follows. For each $i$ define the index set $I_i =
\{2r(i-1) + 1, . . . , 2ri\}$, and let $H_i := U_{I_i}
\textrm{diag}(\vec{\sigma}_{I_i}) V_{I_i}^*$; that is, $H_1$ is the
part of $H_c$ corresponding to the $2r$ largest singular values, $H_2$
is the part corresponding to the next $2r$ largest and so on.  A now
standard computation shows that
\begin{equation}
  \label{eq:usual}
  \sum_{j \ge 2} \|H_j\|_F \le \frac{1}{\sqrt{2r}} \, \|H_c\|_*,  
\end{equation}
and thus
\[
\sum_{j \ge 2} \|H_j\|_F \le \|H_0\|_F + \sqrt{\frac{2}{r}} \|M_c\|_* 
\]
since $\|H_0\|_* \le \sqrt{2r} \, \|H_0\|_F$ by Cauchy-Schwarz.  

Now the restricted isometry property gives
\begin{equation}
\label{eq:detlafour}
(1-\delta_{4r})\|H_0 + H_1\|^2_F \le \|{\cal A}(H_0+H_1)\|_F^2, 
\end{equation}
and observe that 
\[
\|{\cal A}(H_0+H_1)\|_F^2 = \<{\cal A}(H_0 + H_1), {\cal A}(H -
\sum_{j \ge 2} H_j)\>.
\]
We first argue that
\begin{equation}
  \label{eq:AH}
  \<{\cal A}(H_0+H_1), {\cal A}(H)\> \le \|H_0+H_1\|_F \, \sqrt{4r} \|{\cal A}^*{\cal A}(H)\|.   
\end{equation}
To see why this is true, let $U \Sigma V^*$ be the reduced SVD of $H_0
+ H_1$ in which $U$ and $V$ are $n \times r'$, and $\Sigma$ is $r'
\times r'$ with $r' = \rank(H_0+H_1) \le 4r$.  We have
\begin{align*}
  \<{\cal A}(H_0+H_1), {\cal A}(H)\> & = \<H_0+H_1,  {\cal A}^* {\cal A}(H)\> \\
  & = \<\Sigma, U^* [{\cal A}^* {\cal A}(H)] V\> \\
  & \le \|\Sigma\|_F \|U^* [{\cal A}^* {\cal A}(H)] V\|_F\\
  & = \|H_0+H_1\|_F \|U^* [{\cal A}^* {\cal A}(H)] V\|_F.
\end{align*}
The claim follows from $\|U^* [{\cal A}^* {\cal A}(H)] V\|_F \le
\sqrt{r'} \|{\cal A}^*{\cal A}(H)\|$, which holds since $U^* [{\cal
  A}^* {\cal A}(H)] V$ is an $r' \times r'$ matrix with spectral norm
bounded by $\|{\cal A}^*{\cal A}(H)\|$.  Second, Lemma \ref{teo:theta}
implies that for $j \geq 2$
\begin{equation}
  \label{eq:HOHj}
  \<{\cal A}(H_0), {\cal A}(H_j)\> \le \delta_{4r} \|H_0\|_F \, \|H_j\|_F, 
\end{equation}
and similarly with $H_1$ in place of $H_0$. Note that because $H_0$ is
orthogonal to $H_1$, we have that $\|H_0 + H_1\|_F^2 = \|H_0\|_F^2 +
\|H_1\|_F^2$ and thus $\|H_0\|_F + \|H_1\|_F \le \sqrt{2}
\|H_0+H_1\|_F$. This gives
\begin{equation}
  \label{eq:HOHjp}
  \<{\cal A}(H_0+H_1), {\cal A}(H_j)\> \le \sqrt{2} \delta_{4r} \|H_0+H_1\|_F \, \|H_j\|_F. 
\end{equation}

Taken together, \eqref{eq:detlafour}, \eqref{eq:AH} and
\eqref{eq:HOHjp} yield
\begin{align*}
  (1-\delta_{4r}) \|H_0+H_1\|_F & \le \sqrt{4r} \|{\cal A}^*{\cal
    A}(H)\|
  + \sqrt{2} \delta_{4r} \sum_{j \ge 2} \|H_j\|_F\\
  & \le \sqrt{4r} \|{\cal A}^*{\cal A}(H)\| + \sqrt{2} \delta_{4r}
  \|H_0\|_F + \frac{2\delta_{4r}}{\sqrt{r}} \|M_c\|_*.
\end{align*}
To conclude, we have that 
\[ 
\|H_0+H_1\|_F \le C_1 \, \sqrt{4r} \|{\cal
  A}^*{\cal A}(H)\| + C_1 \frac{2\delta_{4r}}{\sqrt{r}} \|M_c\|_*, \quad C_1 = 1/[1-(\sqrt{2}+1)\delta_{4r}], 
\]
provided that $C_1 > 0$.  Our claim \eqref{eq:ds1} then follows from
\eqref{eq:tube} together with
\[
  \|H\|_F  \le \|H_0 + H_1\|_F + \sum_{j \ge 2} \|H_j\|_F \le 2 \|H_0+H_1\|_F + \sqrt{\frac{2}{r}} \|M_c\|_*. 
\]

\subsection{Proof of Theorem \ref{teo:DS1}}
Theorem \ref{teo:DS1} follows by simply plugging $M_r = M$ into
Theorem \ref{teo:nucBound}.  To generalize the results, note that
there are only two requirements on $M, \cA$ and $y$ used in the proof.
\begin{itemize}
\item $\opnorm{\cA^*(\cA(M) - y)} \leq \lambda$
\item $\rank(M) = r$ and $\delta_{4r} < \sqrt{2} -1$.
\end{itemize}
Thus, the steps above also prove the following Lemma which is useful
in proving Theorem \ref{teo:oracle}.
\begin{lemma}
\label{lem:generalM}
Assume that $X$ is of rank at most $r$ and that $\delta_{4r} <
\sqrt{2}-1$. Suppose $\lambda$ obeys $\|{\cal A}^*(y - \cA(X))\| \le
\lambda$. Then the solution $\hat M$ to \eqref{eq:ds} obeys
  \begin{equation}
    \|\hat M - X\|^2_F \le C_0 \, r \lambda^2. 
  \end{equation}
  where $C_0$ is a small constant depending only on the isometry
  constant $\delta_{4r}$. 
\end{lemma}

\subsection{Proof of Theorem \ref{teo:oracle}}
In this section, $\lambda = 16n\sigma^2$ and we take as given that
$\opnorm{\cA^*(z)} \leq \lambda/2$ (and thus, by Lemma
\ref{lem:opNormz}, the end result holds with probability at least $1 -
2e^{-cn}$).  The novelty in this proof -- the way it differs from
analogous proofs in compressive sensing -- is in the use of a middle
estimate $\bar{M}$. Define $K$ as
\begin{equation}
\label{eq:barM}
K(X; M) \equiv \gamma \rank(X) + \twonorm{\cA(X) - \cA(M)}^2, \qquad \gamma = \frac{\lambda^2}{4(1+\delta_1)} 
\end{equation}
and let $\bar{M} = \text{argmin}_X K(X,M)$.  In words, $\bar{M}$
achieves a compromise between goodness of fit and parsimony in the
model with noiseless data. The factor $\gamma$ could be replaced by
$\lambda^2$, but the derivations are cleanest in the present form. We
begin by bounding the distance between $M$ and $\bar{M}$ using the
RIP, and obtain
\begin{equation}
\label{eq:boundMbarM}
 \twonorm{\bar{M} - M}^2 \leq \frac{1}{1 - \delta_{2r}} \twonorm{\cA(\bar{M}) - \cA(M)}^2
\end{equation}
where the use of the isometry constant $\delta_{2r}$ follows from the
fact that $\rank(\bar{M}) \leq \rank(M)$.

We now develop a bound about $\twonorm{\hat{M} - \bar{M}}^2$.  Lemma
\ref{lem:opNorm} gives
\[
\opnorm{\cA^*(y - \cA(\bar{M}))} \leq \opnorm{\cA^*(z)} +
\opnorm{\cA^*\cA(M - \bar{M})} \leq \lambda, 
\]
i.e.~$\bar{M}$ is feasible for \eqref{eq:ds}.  Also, $\rank(\bar{M})
\leq \rank(M)$ and, thus, plugging $\bar{M}$ into Lemma
\ref{lem:generalM} gives
\[
\fronorm{\hat{M} - \bar{M}}^2 \leq C \lambda^2 \rank(\bar{M}).
\]
Combining this with \eqref{eq:boundMbarM} gives
\begin{align}
  \fronorm{\hat{M} - M}^2 &\leq 2 \fronorm{\hat{M} - \bar{M}}^2 + 2 \fronorm{\bar{M} - M}^2 \nonumber\\
  &\leq 2C \lambda^2 \rank(\bar{M}) +  \frac{2}{1 - \delta_{2r}} \twonorm{\cA(\bar{M}) - \cA(M)}^2 \nonumber\\
  &\leq C' K(\bar M; M)
\label{eq:biasVarfBound}
\end{align}
where $C' = \max(8 C(1+\delta_1), 2/(1-\delta_{2r}))$. 

Now $\bar{M}$ is the minimizer of $K(\cdot; M)$, and so $ K(\bar
M; M) \le K(M_0;M)$, where
\begin{equation}
\label{eq:Mnot}
M_0 = \sum_i \,\sigma_i(M) 1_{\{\sigma_i(M) > \lambda\}}\,u_i v_i^*. 
\end{equation}
We have 
\begin{align*}
  K(M_0; M) &\leq \gamma \sum_{i=1}^r 1_{\{\sigma_i(M) > \lambda\}} + \twonorm{\cA(M-M_0)}^2\\
  &\leq \gamma \sum_{i=1}^r 1_{\{\sigma_i(M) > \lambda\}} + (1+\delta_r)\fronorm{M - M_0}^2\\
  &\leq (1+\delta_r)\sum_{i=1}^r \min(\lambda^2, \sigma_i^2(M)).
\end{align*}
In conclusion, the proof follows from $\lambda = 16n\sigma^2$ since
\[
\fronorm{\hat{M} - M}^2 \leq C' \sum_{i=1}^r \min(\lambda^2, 
\sigma_i^2(M)).
\] 

\begin{lemma}
\label{lem:opNorm}
The minimizer $\bar{M}$ obeys
\[
\opnorm{\cA^*\cA(\bar{M} - M)} \leq \lambda/2.
\]
\end{lemma}
\begin{proof}
  Suppose not.  Then there are unit-normed vectors $u,v \in \R^n$
  obeying
\[ 
\<uv^*, \cA^*\cA(\bar{M} - M)\> > \lambda/2. 
\] 
We construct the rank-1 perturbation $M' = \bar{M} - \alpha uv^*$,
$\alpha = \<uv^*, \cA^*\cA(\bar{M} - M)\>/\twonorm{\cA(uv^*)}^2$, and
claim that $K(M':M) < K(\bar{M};M)$ thus providing the
contradiction.  We have 
\begin{align*}
  \twonorm{\cA(M' - M)}^2 & = \twonorm{\cA(\bar M - M)}^2 -2\alpha
  \<\cA(uv^*), \cA(\bar M - M)\> + \alpha^2
  \twonorm{\cA(uv^*)}^2\\
  & = \twonorm{\cA(\bar M - M)}^2 - \alpha^2 \twonorm{\cA(uv^*)}^2.
\end{align*}
It then follows that 
\begin{align*}
  K(M'; M) &\leq \gamma (\rank(M) + 1) + \twonorm{\cA(\bar{M}-M)}^2
  -
  \alpha^2 \twonorm{\cA(uv^*)}^2\\
  &= K(\bar M; M) + \gamma - \alpha^2 \twonorm{\cA(uv^*)}^2.
\end{align*}
However, $\twonorm{\cA(uv^*)}^2 \le (1 + \delta_1) \fronorm{uv^*}^2 =
1 + \delta_1$ and, therefore, $\alpha^2 \twonorm{\cA(uv^*)}^2 >
\gamma$ since $\<uv^*, \cA^*\cA(\bar{M} - M)\> > \lambda/2$. 
\end{proof}

\subsection{Proof of Theorem \ref{teo:instanceOpt}}

Three useful lemmas are established in the course of the proof of this
more involved result, and we would like to point out that these can be
used as powerful error bounds themselves.  Throughout the proof, $C$
is a constant that may depend on $\delta_{4r}$ only, and whose value
may change from line to line.  An important fact to keep in mind is
that under the assumptions of the theorem, $\delta_{4\bar{r}}$ can be
bounded, with high probability, by an arbitrarily small constant
depending on the size of the scalar $c_1$ appearing in the condition
$\bar{r} \leq c_1 m/n$.  This is a consequence of Theorem
\ref{teo:RIP}.  In particular, $\delta_{4\bar{r}} \leq (\sqrt{2} -
1)/2$ with probability at least $1 - De^{-dm}$.

\begin{lemma} 
\label{lem:justRIP}
Let $\bar{M}$ and $M_0$ be defined via \eqref{eq:barM} and
\eqref{eq:Mnot}, and set 
\[
r = \max(\rank(\bar{M}), \rank(M_0)).
\]
Suppose that $\delta_{4r} < \sqrt{2} - 1$ and that $\lambda$ obeys
$\|{\cal A}^*(z)\| \le \lambda/2$. Then the solution $\hat M$ to
\eqref{eq:ds} obeys
  \begin{equation}
    \label{eq:justRIP}
    \|\hat M - M\|^2_F \le C_0 \left( \sum_{i=1}^n \min(\lambda^2, \sigma_i^2(M)) +  \twonorm{\cA(M - M_0)}^2\right),   
  \end{equation}
  where $C_0$ is a small constant depending only on the isometry
  constant $\delta_{4r}$.
\end{lemma}
\begin{proof}
  The proof is essentially the same as that of Theorem
  \ref{teo:oracle}, and so we quickly go through the main steps.  Set
  $M_c = M - M_0$ so that $M_c$ only contains the singular values
  below the noise level.  First,
\begin{align*}
  \fronorm{\bar{M} - M}^2 &\leq 2\fronorm{\bar{M} - M_0}^2 + 2\fronorm{M_c}^2\\
  &\leq \frac{2}{1 - \delta_{2r}}\twonorm{\cA(\bar{M} - M_0)}^2 + 2\fronorm{M_c}^2\\
  &\leq \frac{4}{1 - \delta_{2r}} \twonorm{\cA(\bar{M} - M)}^2 +
  \frac{4}{1 - \delta_{2r}} \twonorm{\cA(M_c)}^2 + 2\fronorm{M_c}^2.
\end{align*}
Second, we bound $\fronorm{\hat{M} - \bar{M}}$ using the exact same
steps as in the proof of Theorem \ref{teo:oracle}, and obtain
\[
\fronorm{\hat{M} - \bar{M}}^2 \leq C r \lambda^2.
\]
Hence, 
\[
\fronorm{\hat{M} - M}^2 \leq C (K(\bar{M}; M) + \twonorm{\cA(M_c)}^2 +
\fronorm{M_c}^2). 
\]
Finally, use $K(\bar{M}; M) \leq K(M_0; M)$ as before, and simplify to
attain \eqref{eq:justRIP}.
\end{proof}

The factor $\fronorm{\cA(M_c)}^2$ in \eqref{eq:justRIP} prevents us
from stating the bound as the near-ideal bias-variance-trade-off
\eqref{eq:fullBiasVar}.  However, many random measurement ensembles
obeying the RIP are also unlikely to drastically change the norm of
\textit{any} fixed matrix (see \eqref{eq:concentration}). Thus, we
expect that $\twonorm{\cA(M_c)} \approx \twonorm{M_c}$ with high
probability.  Specifically, if $\cA$ obeys \eqref{eq:concentration},
then
\begin{equation}
\label{eq:smallIncrease}
\twonorm{\cA(M_c)}^2 \leq 1.5 \fronorm{M_c}^2
\end{equation}
with probability at least $1 - De^{-cm}$ for fixed constants $D,c$.
An important point here is that this inequality only holds (with high
probability) when $M_c$ is fixed, and $\cA$ is chosen randomly
(independently).  In the worst-case-scenario, one could have
\[\twonorm{\cA(M_c)} = \opnorm{\cA}\fronorm{M_c}\]
where $\opnorm{\cA}$ is the operator norm of $\cA$.  Thus we emphasize
that the bound holds with high probability for a given $M$ verifying
our conditions, but may not hold uniformly over all such $M$'s.

Returning to the proof, \eqref{eq:smallIncrease} together with
\[
\fronorm{M_c}^2 = \sum_{i=1}^n \sigma_i^2(M) 1_{\{\sigma_i(M) < \lambda\}}
\]
give the following lemma:
\begin{lemma} 
\label{lem:RIPinst}
Fix $M$ and suppose $\cA$ obeys \eqref{eq:concentration}. Then under
the assumptions of Lemma \ref{lem:justRIP}, 
the solution $\hat M$ to \eqref{eq:ds} obeys
  \begin{equation}
    \|\hat M - M\|^2_F \le C_0 \sum_{i=1}^n \min(\lambda^2, \sigma_i^2(M))  
  \end{equation}
  with probability at least $1 - D e^{-cn}$ where $C_0$ is a small
  constant depending only on the isometry constant $\delta_{4r}$, and
  $c$, $D$ are fixed constants.
\end{lemma}

The above two lemmas require a bound on the rank of $M_0$.  However,
as the noise level approaches zero, the rank of $M_0$ approaches the
rank of $M$, which can be as large as the dimension.  This requires
further analysis, and in order to provide theoretical error bounds
when the noise level is low (and $M$ has full rank, say), a certain
property of many measurement operators is useful.  We call it the NNQ
property, and is inspired by a similar property from compressive
sensing, see \cite{wojtaszczyk08}.
\begin{definition}[NNQ]
  Let $B_*^{n\times n}$ be the set of $n \times n$ matrices with
  nuclear norm bounded 1.  Let $B_{\ell_2}^m$ be the standard $\ell_2$
  unit ball for vectors in $\R^m$.  We say that $\cA$ satisfies
  NNQ($\alpha$) if
\begin{equation}
\label{eq:NNQ}
\cA(B_*^{n\times n}) \supseteq \alpha B_{\ell_2}^m.
\end{equation}
\end{definition}
This condition may appear cryptic at the moment.  To give a taste of
why it may be useful, note that Lemma \ref{teo:nucBound} includes
$\nucnorm{M - M_r}$ as part of the error bound.  The point is that
using the NNQ condition, we can find a proxy for $M - M_r$, which we
call $\tilde{M}$, satisfying $\cA(\tilde{M}) = \cA(M - M_r)$, but also
$\nucnorm{\tilde{M}} \leq \twonorm{\cA(M - M_r)}/\alpha$.  Before
continuing this line of thought, we prove that Gaussian measurement
ensembles satisfy $\text{NNQ}(\mu \sqrt{n/m})$ with high probability
for some fixed constant $\mu > 0$.

\begin{theorem}[NNQ for Gaussian measurements]
  Suppose $\cA$ is a Gaussian measurement ensemble and $m \leq C
  n^2/\log(m/n)$ for some fixed constant $C > 0$.  Then $\cA$
  satisfies $\text{NNQ}(\mu \sqrt{n/m})$ with probability at least $1
  - 3e^{-cn}$ for fixed constants $c$ and $\mu$.
\end{theorem}
\begin{proof}
  Put $\alpha = \mu \sqrt{n/m}$ and suppose $\cA$ does not satisfy
  $\text{NNQ}(\alpha)$.  Then there exists a vector $x \in \R^m$ with
  $\twonorm{x} = 1$ such that
  \[ \<\cA(M), x\> \leq \alpha \quad \text{ for all } M \in
  B_*^{n\times n}.\] 
  In particular,
\[
\opnorm{\cA^*(x)} \leq \alpha.
\]
Let $\bar{B}_{\ell_2}^m \subset B_{\ell_2}^m$ be an $\alpha$-net for $B_{\ell_2}^m$ with
$|\bar{B}_{\ell_2}^m| \leq (3/\alpha)^m$.  Then there exists $\bar{x}
\in \bar{B_{\ell_2}^m}$ with $\twonorm{\bar{x} - x} \leq \alpha$
satisfying
\[
\opnorm{\cA^*(\bar{x})} \leq \opnorm{\cA^*(\bar{x} - x)} +
\opnorm{\cA(x)} \leq \<uv^*, \cA^*(\bar{x} - x)\> + \alpha,
\]
where $u,v$ are the left and right singular vectors of $\cA^*(\bar{x}
- x)$ corresponding to the top singular value.  Then
\[
\<uv^*, \cA^*(\bar{x} - x)\> = \<\cA(uv^*),\bar{x} - x\> \leq
\twonorm{\cA(uv^*)} \, \twonorm{\bar{x} - x} \leq \sqrt{1 +
  \delta_1} \, \alpha
\] 
and, therefore,
\[
\opnorm{\cA^*(\bar{x})} \leq 3\alpha
\] 
assuming $\delta_1 \leq 1$ (this occurs with probability at least $1 -
2e^{-cn}$ when $m \geq Cn$ for fixed constants $c,C$).

We will provide the contradiction by showing that with high
probability, $\opnorm{\cA^*(\bar{x})} > 3 \alpha$, for all $\bar{x}
\in \bar{B}_*^{n\times n}$.  For each $\bar{x}$,
$\cA^*(\bar{x})$ is equal in distribution to $\frac{1}{\sqrt{m}} Z$, where $Z$
is a matrix with i.i.d.~standard normal entries.  Let $Z_i$ be the $i$th
column of $Z$.  Then
\begin{align*}
  \P(\opnorm{\cA^*(\bar{x})} \leq 3 \alpha) &\leq \P(\opnorm{Z} \leq 3 \sqrt{m} \alpha)\\
  &\leq \P(\max_{i=1,\ldots,n} \twonorm{Z_i} \leq 3 \sqrt{m} \alpha);
\end{align*}
the second step uses the fact that the operator norm of $Z$ is always
larger or equal to the $\ell_2$ norm of any column.  With $\alpha =
\mu \sqrt{n/m}$ and using the fact that the columns are independent,
this yields
\[
\P(\opnorm{\cA^*(\bar{x})} \leq 3 \alpha) \leq \P(\twonorm{Z_1}^2 \leq
9 \mu^2 n)^n.
\] 
However, $\twonorm{Z_1}^2$ is a chi-squared random variable with $n$
degrees of freedom, and can be bounded using a standard concentration
of measure result \cite{Laurent00}:
\[
\P(\twonorm{Z_1}^2 - n \leq -t \sqrt{2n}) \leq e^{-t^2/2}.
\]
Hence,
\[
\P(\opnorm{\cA^*(\bar{x})} \leq 3 \alpha) \leq e^{-cn^2},
\]
where $c = (1 - 9 \mu^2)^2/4$ (we require $\mu < 1/3$ here).  Thus, by
the union bound,
\[
\P\left(\min_{\bar{x} \in \bar{B_{\ell_2}^m}} \opnorm{\cA^*(\bar{x})}
  \leq 3 \alpha\right) \leq (3/\alpha)^m e^{-cn^2} = \exp\left(m
  \log\left(\frac{3 \sqrt{m}}{\mu \sqrt{n}}\right) - cn^2\right) \leq
e^{-c' n^2}
\] 
provided that $m \leq C n^2/\log(m/n)$ for fixed constants, $C, c'$.
The theorem is established.
\end{proof}

Note that the preceding proof can be repeated when $\cA$ is a
sub-Gaussian measurement ensemble; the only difference is that $Z$
above will contain sub-Gaussian entries, rather than Gaussian entries.

Using the NNQ property, we can now bound the error when the noise
level is low; this does not involve any condition on the rank of
$M_0$, and does not involve a term in the bound depending on
$\nucnorm{M - M_0}$.  
\begin{lemma}
\label{lem:lowNoise}
Suppose that $\cA$ satisfies NNQ($\mu \sqrt{n/m}$) for a fixed
constant $\mu$ and that $\opnorm{\cA^*(z)} \leq \lambda$.  Let $\bar{r} \geq c m/n$ for some fixed numerical constant $c$, and suppose that  that $\delta_{4\bar r} \leq \frac{1}{2} (\sqrt{2} - 1)$.  Let 
\[M_{\bar{r}} = \sum_{i=1}^{\bar{r}} \sigma_i(M)u_i v_i^*.\]
Let $\hat{M}$ be the solution to $\eqref{eq:ds}$.  
Then
\begin{equation}
  \fronorm{\hat{M} - M} \leq C(\lambda \sqrt{\bar{r}} + 
\twonorm{\cA(M - M_{\bar{r}})}) + \fronorm{M - M_{\bar{r}}}.
\end{equation}
\end{lemma}
\begin{proof}
  Set $M_c = M - M_{\bar{r}} = \sum_{i=\bar{r}+1}^n \sigma_i(M) u_i
  v_i^*$. The $\text{NNQ}(\alpha)$ property with $\alpha = \mu
  \sqrt{n/m}$ gives
\[
\cA(M_c) = \cA(\tilde{M})
\] 
for some $\tilde{M}$ obeying $\nucnorm{\tilde{M}} \leq
\twonorm{\cA(M_c)}/\alpha$. We also take note of the identity
$\cA(M_{\bar{r}} + \tilde{M}) = \cA(M)$.  It follows from Lemma
\ref{teo:nucBound} that
\[
\fronorm{\hat{M} - (M_{\bar{r}} + \tilde{M})} \leq C (\lambda
\sqrt{\bar{r}} + \nucnorm{\tilde{M}}/\sqrt{\bar{r}}).
\]
Plugging in $\nucnorm{\tilde{M}} \leq
\twonorm{\cA(M_c)}/\alpha$, 
along with $\bar{r} \geq c m/n$, we obtain
\[
\fronorm{\hat{M} - (M_{\bar{r}} + \tilde{M})} \leq C(\lambda
\sqrt{\bar{r}} + \twonorm{\cA(M_c)}).
\]
Therefore, 
\begin{equation}
\label{eq:instBound}
\fronorm{\hat{M} - M} \leq C(\lambda \sqrt{\bar{r}} + \twonorm{\cA(M_c)}) + \fronorm{\tilde{M}} + \fronorm{M_c}.
\end{equation}

It remains to bound $\fronorm{\tilde{M}}$.  As in the proof of Lemma
\ref{teo:nucBound}, decompose $\tilde{M}$ as $\tilde{M} = \tilde{M}_1
+ \tilde{M}_2 + \ldots$ so that $\tilde{M}_1$ corresponds to the
largest $\bar{r}$ singular values of $\tilde{M}$, $\tilde{M}_2$
corresponds with the next $\bar{r}$ largest, and so on.  Just as
before,
\[
\fronorm{\tilde{M}} \leq \sum_i \fronorm{\tilde{M}_i} \leq
\fronorm{\tilde{M}_1} + \nucnorm{\tilde{M}}/\sqrt{\bar{r}}.
\] 
We now bound $\fronorm{\tilde{M}_1}$.  By the RIP,
\begin{align*}
  \fronorm{\tilde{M}_1} &\leq
  \frac{1}{\sqrt{1 - \delta_{\bar{r}}}} \twonorm{\cA(\tilde{M}_1)}\\
  & = \frac{1}{\sqrt{1 - \delta_{\bar{r}}}}\Bigl(
  \twonorm{\cA(\tilde{M}) - \sum_{i\geq 2} \cA(\tilde{M}_i)}\Bigr)\\
  &\leq \frac{1}{\sqrt{1 -
      \delta_{\bar{r}}}}\Bigl(\twonorm{\cA(\tilde{M})} + \sum_{i\geq
    2} \twonorm{\cA(\tilde{M}_i)}\Bigr).
\end{align*}
By the RIP again, $\twonorm{\cA(\tilde{M}_i)} \leq \sqrt{1 +
  \delta_{\bar{r}}} \fronorm{\tilde{M}_i}$, and so
\[
\sum_{i\geq 2} \twonorm{\cA(\tilde{M}_i)} \leq \sqrt{1 +
  \delta_{\bar{r}}} \, \sum_{i\geq 2} \fronorm{\tilde{M}_i} \leq \sqrt{1
  + \delta_{\bar{r}}} \, \frac{\nucnorm{\tilde{M}}}{\sqrt{\bar{r}}}.
\] 
This together with $\cA(\tilde{M}) = \cA(M_c)$
give 
\[
\fronorm{\tilde{M}} \leq \sqrt{\frac{1 + \delta_r}{1 -
    \delta_r}}\Bigl(\twonorm{\cA(M_c)} +
\frac{\nucnorm{\tilde{M}}}{\sqrt{\bar{r}}}\Bigr).
\] 
However, $\nucnorm{\tilde{M}} \leq \twonorm{\cA(M)}/\alpha \leq
\sqrt{\bar{r}} \twonorm{\cA(M)}/(\mu \sqrt{c})$ and, therefore,
\[
\fronorm{\tilde{M}} \leq C \twonorm{\cA(M_c)}.
\]
Inserting this into \eqref{eq:instBound} completes the proof of the
lemma.
\end{proof}

We are now in position to prove our main theorem concerning the
recovery of matrices with decaying singular values (Theorem
\ref{teo:instanceOpt}).  There are three cases to consider depending
on the number of singular values of $M$ standing above the noise
level.  In each case, we need the inequality
\begin{equation}
\label{eq:smallIncrease}
\fronorm{\cA(M_c)}^2 \leq 1.5 \fronorm{M_c}^2
\end{equation}
which holds with probability at least $1 - De^{-cn}$ for any measurement
ensemble satisfying $\eqref{eq:concentration}$ (including the Gaussian measurement ensemble).  Put $\lambda =
16n\sigma^2$ and recall the definition of $M_0$:
\[
M_0 = \sum_{i=1}^n \sigma_i(M) 1_{\{\sigma_i(M) \geq \lambda\}} \, u_i
v_i^*
\]
whose rank is exactly the number of singular values of $M$ above the
noise level.  There are three cases to consider depending mostly on the interplay between the singular values of $M$ and the noise level.

\subsection*{Case 1: high noise level}
Suppose $K(M_0; M) \leq \frac{\lambda^2}{4(1+\delta_1)} \bar{r}$.
Then $\rank(M_0) \leq \bar{r}$ and $\rank(\bar{M}) \leq \bar{r}$ by
definition of $\bar M$.  Hence, Lemma \ref{lem:RIPinst} gives
\[
\fronorm{\hat{M} - M}^2 \leq C \sum_{i=1}^n \min(n\sigma^2,
\sigma_i^2(M))
\] 
with probability at least $1 - 2e^{-cn}$.

\subsection*{Case 2: low noise level}
Suppose $K(M_0; M) > \frac{\lambda^2}{4(1+\delta_1)} \bar{r}$ and
$\rank(M_0) \geq \bar{r}$.  It follows from \eqref{eq:concentration}
that
\begin{equation}
\label{eq:aBound}
\twonorm{\cA(M - M_{\bar{r}})}^2 \leq \sqrt{1.5} \fronorm{M - M_{\bar{r}}}^2
\end{equation}
with probability at least $1 - De^{-cn}$.  Now, for the Gaussian measurement ensemble, the requirements of Lemma \ref{lem:lowNoise} are met with probability at least $1 - C e^{-cn}$.  Combining \eqref{eq:aBound} with Lemma \ref{lem:lowNoise} yields
\[\fronorm{\hat{M} - M} \leq C (\lambda \sqrt{\bar{r}} + \fronorm{M - M_{\bar{r}}})\]
and thus
\[
\fronorm{\hat{M} -M}^2 \leq 2 C^2 (\lambda^2 \bar{r} + \fronorm{M -
  M_{\bar{r}}}^2) = 2 C^2 \left(\sum_{i=1}^{\bar{r}} \min(\lambda^2,
  \sigma_i^2(M)) + \sum_{i=\bar{r} + 1}^n \sigma_i^2(M)\right).
\]
Since $\lambda = 16n\sigma^2$, this is \eqref{eq:instBound}.

\subsection*{Case 3: medium noise level}
Suppose $K(M_0; M) > \frac{\lambda^2}{4(1 + \delta_1)} \bar{r}$ and
$\rank(M_0) < \bar{r}$.  As in Case 2, we have
\[
\fronorm{\hat{M} -M}^2 \leq 2 C^2 (\lambda^2 \bar{r} + \fronorm{M -
  M_{\bar{r}}}^2).
\]
From $\lambda^2 \bar{r} < 4(1+ \delta_1) K(M_0; M)$, it follows that
\[
\fronorm{\hat{M} -M}^2 \leq 2 C^2 (\lambda^2\rank(M_0) + 4(1 +
\delta_1) \twonorm{\cA(M - M_0)}^2 + \fronorm{M - M_{\bar{r}}}^2).
\]
We also have $\twonorm{\cA(M-M_0)}^2 \leq 1.5 \fronorm{M - M_0}^2$
with probability at least $1 - De^{-cn}$.  Inserting this bound into
the previous equation, along with $\fronorm{M - M_{\bar{r}}} \leq
\fronorm{M - M_0}$, gives the desired conclusion.

These three cases comprise all possibilities. In short, the proof of
Theorem \ref{teo:instanceOpt} is complete.

\subsection{Extension of proofs to the solution to the Lasso \eqref{eq:lasso}}
\label{sec:extension}

In the sparse regression setup, Bickel et al.~\cite{Bickel07} showed
that the Dantzig Selector and the Lasso have analogous properties,
leading to analogous error bounds.  The analogies still hold in the
low-rank matrix recovery problem (for similar reasons).  In fact, all
of the theorems above also hold for the solution to
\eqref{eq:lasso} aside from a shift in those constants appearing in
the assumptions, and those appearing in the error bounds.  To see
this, note that our proofs only used two crucial properties about
$\hat{M}$:
\begin{enumerate}
\item $\nucnorm{\hat{M}} \leq \nucnorm{M}$, 
\item $\opnorm{\cA^*(\cA(\hat{M}) - y)} \leq \lambda$.
\end{enumerate}
The second property automatically holds for the solution to
\eqref{eq:lasso} (but with $\lambda$ replaced by $\mu$). This follows
from the optimality conditions which states that $\cA^*(y - \cA(\hat
M)) \in \partial \|\hat M\|_*$ where $\|\hat M\|_*$ is the family of
subgradients to the nuclear norm at the minimizer. Formally, let
$U\Sigma V^*$ be the SVD of $\hat M$, then
\[
\cA^*(y - \cA(\hat M)) =  \lambda(UV^* + W)
\]
for some $W$ obeying $\opnorm{W} \leq 1$ and $U^* W = 0, WV = 0$ (see
e.g.~\cite{CR08}).  Hence, the second property follows from
$\opnorm{UV^* + W} \leq 1$.

The first property does not necessarily hold for the matrix Lasso, but
a close enough approximation is verified (this is analogous to an
argument made in \cite{Bickel07}).  Suppose that $\opnorm{\cA^*(z)}
\leq c_0 \mu$ for a small constant $c_0$ (which, by Lemma
\ref{lem:opNormz}, holds with high probability for Gaussian noise if
$\mu = Cn\sigma^2$).  Then since $\hat{M}$ minimizes \eqref{eq:lasso},
we have
\[
\frac{1}{2} \twonorm{\cA(\hat{M}) - y}^2 + \mu \nucnorm{\hat{M}} \leq
\frac{1}{2} \twonorm{\cA(M) - y}^2 + \mu \nucnorm{M}.
\] 
Plug in $y = \cA(M) + z$ and rearrange terms to give
\[
\mu \nucnorm{\hat{M}} \le \frac{1}{2} \twonorm{\cA(\hat{M} - M)}^2 + \mu
\nucnorm{\hat{M}} \leq \<\hat{M} - M, \cA^*(z)\> + \mu \nucnorm{M}.
\] 
Since the nuclear norm and the operator norm are dual to each other,
we have $\<\hat{M} - M, \cA^*(z)\> \leq \nucnorm{\hat{M} - M} \cdot
\opnorm{\cA^*(z)} \leq c_0 \mu \nucnorm{H}$, where we use the notation
$H = \hat{M} - M$ as in the proof of Lemma \ref{teo:nucBound}. This
gives
\[
\nucnorm{\hat{M}} \leq c_0 \nucnorm{H} + \nucnorm{M},
\]
which nearly is the first property.  When $c_0$ is chosen to be a
small constant, this factor has no essential detrimental effects on
the proof.  In particular, \eqref{eq:cone} in the proof of Lemma
\ref{teo:nucBound} is replaced by
\[
(1-c_0) \nucnorm{H_c} \leq (1+c_0) \nucnorm{H_0} + 2 \nucnorm{M_c}.
\]
In particular, for $c_0 = 1/2$,
\[
\nucnorm{H_c} \leq 3 \nucnorm{H_0} + 4 \nucnorm{M_c}.
\]
The rest of the proofs follow.

\subsection{Proof of Theorem \ref{teo:minimax}}

We begin with a well-known lemma which gives the minimax risk for
estimating the vector $x \in \R^n$ from the data $y \in \R^m$ and the
linear model
\begin{equation}
  \label{eq:linmodel}
  y = Ax + z, 
\end{equation}
where $A \in \R^{m\times n}$ and the $z_i$'s are i.i.d.~${\cal
  N}(0,\sigma^2)$. 
\begin{lemma}
\label{teo:meanisminimax}
Let $\lambda_i(A^*A)$ be the eigenvalues of the matrix $A^*A$. Then
  \begin{equation}
    \label{eq:meanisminimax}
    \inf_{\hat x} \sup_{x \in \R^n} \E \|\hat x - x\|^2 = \sigma^2 \trace((A^*A)^{-1})= \sum_i
    \frac{\sigma^2}{\lambda_i(A^*A)}. 
  \end{equation}
  In particular, if one of the eigenvalues vanishes (as in the case in
  which $m < n$), then the minimax risk is unbounded.
\end{lemma}
\begin{proof}
  Suppose first that $A$ is the identity matrix. Then it is well known
  that the minimax risk is $n\sigma^2$ and is achieved by $\hat x =
  y$. To see this, recall that for any prior on $x$, the minimax risk
  is lower bounded by the Bayes risk. Consider then the prior which
  assumes that all the components of $x$ are i.i.d.~${\cal
    N}(0,\tau^2)$. Then the Bayes' estimator for this prior is the
  shrinkage estimate given by
\[
\hat{x}_i = \E(x_i | y_i) = \frac{\tau^2}{\tau^2+\sigma^2} y_i, 
\]
and the Bayes risk is
\[
\sum_{i = 1}^n \E(\hat{x}_i - x_i)^2 = n \sigma^2 \,
\frac{\tau^2}{\tau^2 + \sigma^2}.
\]
Clearly, as $\tau \goto \infty$, the lower bound on the minimax risk
goes to $n \sigma^2$. Since this quantity is the risk of the
maximum-likelihood estimate $y$, this proves the claim. Note that by a
simple rescaling argument, this also proves that the minimax risk for
estimating $x$ from $y_i = d_i x_i + z_i$ is $\sum_{i = 1}^n 1/d^2_i$.

We can now prove \eqref{eq:meanisminimax}. We will assume that $m \ge
n$ for simplicity since for $m < n$, the minimax risk is
unbounded. Let $U\Sigma V^*$ be the SVD of $A$, where $U$ is $m \times
n$, $\Sigma$ is $n \times n$ and $V$ is $n \times n$. All the
information about $x$ is in $U^*y$, and so we may just assume that the
data is given by 
\[
y' = U^* y = \Sigma V^* x + U^*z.
\]
Now $z' = U^*z$ is a Gaussian vector with i.i.d.~${\cal
  N}(0,\sigma^2)$ components. Further, set $x' = V^* x$. Since $V$ is
an orthogonal matrix, the minimax risk for estimating $x$ or $x'$ is
the same and, therefore, our problem is that of computing the minimax
risk for estimating $x'$ from
\[
y' = \Sigma x' + z'.
\]
Since $\Sigma$ is a diagonal matrix with diagonal elements
$\sqrt{\lambda_i(A^*A)}$, our previous result applies and establishes
\eqref{eq:meanisminimax}. 
\end{proof}

We are now in position to prove Theorem \ref{teo:minimax}. The set of
rank-$r$ matrices is (much) larger than the set of matrices of the
form
\[
M = U R,
\]
where $U$ is a {\em fixed} orthogonal $n \times r$ matrix with
orthonormal columns (note that the matrices of this form have a fixed
$r$-dimensional column space). Thus,
\[
\inf_{\hat{M}} \sup_{M : \rank(M) = r} \E\|\hat M - M\|_F^2 \ge
\inf_{\hat{M}} \sup_{M : M = UR} \E\|\hat M - M\|_F^2. 
\]
Knowing that $M = UR$ for some unknown $r \times n$ matrix $R$, one
can of course limit ourselves to estimators of the form $\hat M = U
\hat R$, and since
\[
\E \|\hat M - M\|_F^2 = \E \|U\hat R - UR\|_F^2 =  \E \|\hat R - R\|_F^2,  
\]
the minimax risk is lower bounded that by of estimating $R$ from the data 
\[
y = {\cal A}_U(R) + z, 
\]
where ${\cal A}_U$ is the linear map \eqref{eq:AU}.  We then apply
Lemma \ref{teo:meanisminimax} to conclude that the minimax rate is
lower bounded by
\[
\sum_{i} \frac{\sigma^2}{\lambda_i({\cal A}_U^* {\cal A}_U)}. 
\]
The claim follows from the simple lemma below. 

\begin{lemma}
  \label{teo:AU}
  Let $U$ be an $n \times r$ matrix with orthonormal columns. Then all
  the eigenvalues of ${\cal A}_U^* {\cal A}_U$ belong to the interval
  $[1-\delta_r, 1 + \delta_r]$. 
\end{lemma}
\begin{proof}
By definition,
\[
\lambda_{\text{min}}({\cal A}_U^* {\cal A}_U) = \inf_{\|R\|_F \le 1} \<R,
{\cal A}_U^* {\cal A}_U (R)\>
\]
and similarly for $\lambda_{\text{max}}({\cal A}_U^* {\cal A}_U)$ with
a $\sup$ in place of $\inf$. Since
\[
\<R, {\cal A}_U^* {\cal A}_U (R)\> = \| {\cal A}_U (R)\|^2_{\ell_2} =
\|{\cal A}(UR)\|^2,  
\]
the claim follows from 
\[
(1-\delta_r) \|UR\|_F^2 \le \|{\cal A}(UR)\|^2 \le (1+\delta_r) \|UR\|_F^2, 
\]
which is valid since $\rank(UR) \le r$ together with $\|UR\|_F^2 =
\|R\|_F^2$.
\end{proof}

\section{Discussion}
\label{sec:conclusion}
Using RIP-based analysis, this paper has shown that low-rank matrices
can be stably recovered via nuclear-norm minimization from nearly the
minimal possible number of linear samples.  Further, the error bound
is within a constant of the expected minimax error, and of an expected
oracle error, and extends to the case when $M$ has full rank.

This work differs from the main thrust of the recent literature on
low-rank matrix recovery, which has concentrated on the `RIPless'
matrix completion problem.  An interesting observation regarding
matrix completion is that when the measurements are randomly chosen
entries of $M$, one requires at least about $n r \log n$ measurements to
recover $M$ \textit{by any method} when $\rank(M) = O(1)$
\cite{CR08, CT09}. In contrast, this paper shows that on the order of
$n r$ measurements are enough provided these are sufficiently random.

The popularity of the matrix completion model stems from the fact that
this setup currently dominates the applications of low-rank matrix
recovery.  There are far fewer applications in which the measurements
are random linear combinations of many entries of $M$ (quantum-state
tomography is a notable application though).  As a great deal of
attention is given to low-rank matrix modeling these days, with new
applications being discovered all the time, this may change rapidly.
We hope that our theory encourages further applications and research
in this direction.

\small
\bibliographystyle{plain}
\bibliography{ref}

\end{document}